\documentclass[a4paper,12pt]{article}

\usepackage{complexity} 
\usepackage{enumerate}
\usepackage{amsmath, amssymb, amsthm,complexity}
\usepackage{euler}
\usepackage{authblk}
\usepackage{todonotes}

\usepackage{mathtools}
\usepackage{thmtools}
\usepackage{thm-restate}
\usepackage{hyperref}
\usepackage{cleveref}
\usepackage{complexity}
\usepackage{authblk}
\theoremstyle{plain}
\usepackage{todonotes}
\usepackage{microtype}
\usepackage{multirow}
\usepackage{amsmath,amssymb}
\usepackage{comment}
\usepackage{enumerate}
\usepackage{xspace}
\usepackage{listings}
\usepackage{color}
\usepackage{cite}
\usepackage{graphicx}
\RequirePackage{fancyhdr}
 \usepackage{xcolor}
\usepackage{boxedminipage}
\usepackage{algorithm2e}

\newcommand{\defparproblem}[4]{
  \vspace{3mm}
\noindent\fbox{
  \begin{minipage}{.95\textwidth}
  \begin{tabular*}{\textwidth}{@{\extracolsep{\fill}}lr} \textsc{#1}\\ \end{tabular*}
  {\bf{Input:}} #2  \\
  {\bf{Parameter:}} #3 \\
  {\bf{Question:}} #4
  \end{minipage}
  }
  \vspace{2mm}
}

\def\mypara#1{\sbox0{\parbox{\linewidth}{%
  \parindent0pt
 #1\par\xdef\myparasize{\the\prevgraf}}}%
\ifnum\myparasize=1
{\parindent0pt #1\par}%
\else
#1
\fi}

\newcommand{\AAA}{{\mathcal A}}

\newcommand{\OO}{{\mathcal O}}

\newcommand{\FF}{{\mathcal F}}

\newcommand{\II}{{\mathcal I}}

\newcommand{\ccP}{{\mathcal P}}

\newcommand{\pionetopiddeletion}{{\sc Finite ($\Pi_1, \Pi_2, \dotsc , \Pi_d$) Vertex Deletion}}
\newcommand{\gpionetopiddeletion}{{\sc ($\Pi_1, \Pi_2, \dotsc , \Pi_d$) Vertex Deletion}}
\newcommand{\indpionetopiddeletion}{{\sc Individually tractable ($\Pi_1, \Pi_2, \dotsc , \Pi_d$) Vertex Deletion}}
\newcommand{\pionetopidmodulator}{($\Pi_1, \Pi_2, \dotsc , \Pi_d$)-modulator }

\newcommand{\disjointpionetopiddeletioncomp}{{\sc Disjoint Finite ($\Pi_1, \Pi_2, \dotsc , \Pi_d$)-VDC}}
\newcommand{\disjointpionetopiddeletioncompu}{{\sc Disjoint Finite ($\Pi_1, \Pi_2, \dotsc , \Pi_d$)-VDC with Undeletable Vertices}}

\newcommand{\sv}[1]{}

\newtheorem{reduction rule}{Reduction Rule}
\newtheorem{branching rule}{Branching Rule}
\newtheorem{definition}{Definition}

\newtheorem{lemma}{Lemma}

\newtheorem{proposition}{Proposition}
\newtheorem{theorem}{Theorem}
\newtheorem{claim}{Claim}
\usepackage{listings}
\usepackage{color}
\usepackage{cite}
\usepackage{complexity}
\usepackage{todonotes}

\title{Deletion to Scattered Graph Classes I - case of finite number of graph classes\footnote{Preliminary version of the paper appeared in proceedings of IPEC 2020}}

\date{}
\author[1]{Ashwin Jacob}
\author[2]{Jari J. H. de Kroon}
\author[3]{Diptapriyo Majumdar}
\author[1]{Venkatesh Raman}
\affil[1]{The Institute of Mathematical Sciences, HBNI, Chennai, India
  \texttt{\{ajacob|vraman\}@imsc.res.in}}
\affil[2]{Eindhoven University of Technology, The Netherlands
    \texttt{j.j.h.d.kroon@tue.nl}}
\affil[3]{Indraprastha Institute of Information Technology Delhi, India
    \texttt{diptapriyo@iiitd.ac.in}}

\begin{document}

\maketitle

\begin{abstract}
Graph-modification problems, where we modify a graph by adding or deleting vertices or edges or contracting edges to obtain a graph in a {\it simpler} class, is a well-studied optimization problem
in all algorithmic paradigms including classical, approximation and parameterized complexity. Specifically, graph-deletion problems, where one needs to delete a small number of vertices to make the resulting graph to belong to a given non-trivial hereditary graph class, captures several well-studied problems including {\sc Vertex Cover}, {\sc Feedback Vertex Set}, {\sc Odd Cycle Transveral}, {\sc Cluster Vertex Deletion}, and {\sc Perfect Deletion}. Investigation into these problems in parameterized complexity has given rise to powerful tools and techniques.

We initiate a study of a natural variation of the problem of deletion to {\it scattered graph classes}. We want to delete at most $k$ vertices so that in the resulting graph, each connected component belongs to one of a constant number of graph classes. 
As our main result, we show that this problem is fixed-parameter tractable (FPT) when the deletion problem corresponding to each of the finite number of graph classes is known to be FPT and the properties that a graph belongs to any of the classes is expressible in Counting Monodic Second Order (CMSO) logic. While this is shown using some black box theorems in parameterized complexity, we give a
faster FPT algorithm when each of the graph classes has a finite forbidden set.  

\end{abstract}

\newpage


\section{Introduction}
\label{sec:intro}
Graph modification problems, where we want to modify a given graph by adding or deleting vertices or edges to obtain a {\it simpler} graph are well-studied problems in algorithmic graph theory.
A classical work of Lewis and Yannakakis~\cite{LewisY80} (see also~\cite{Yannakakis81a}) showed the problem {\sc NP-complete} for the resulting simpler graph belonging to any non-trivial hereditary graph class. A graph property is simply a collection of graphs and is non-trivial if the class and its complement contain infinitely many graphs. A graph class is hereditary if it is closed under induced subgraphs.  Since the work of Lewis and Yannakakis, the complexity of the problem has been studied in several algorithmic paradigms including approximation and parameterized complexity.
Specifically, deleting at most $k$ vertices to a fixed hereditary graph class is an active area of research in parameterized complexity over the last several years yielding several powerful tools and techniques.  Examples of such problems include {\sc Vertex Cover}, {\sc Cluster Vertex Deletion}, {\sc Feedback Vertex Set} and {\sc Chordal Vertex Deletion}.

It is well known that any hereditary graph class can be described by a forbidden set of graphs, finite or infinite, that contains all minimal forbidden graphs in the class.
In parameterized complexity, it is known that the deletion problem is fixed-parameter tractable ({FPT}) as long as the resulting hereditary graph class has a finite forbidden set~\cite{Cai96}. This is shown by an easy reduction to the {\sc Bounded Hitting Set} problem. This includes, for example, deletion to obtain a split graph or a cograph. We also know {FPT} algorithms for specific graph classes defined by infinite forbidden sets like {\sc Feedback Vertex Set} and {\sc Odd Cycle transversal}~\cite{cygan2015parameterized}.
While the precise characterization of the class of graphs for which the deletion problem is FPT is elusive, there are graph classes for which the problem is W-hard~\cite{Lokshtanov08-wheel,HHJKV13}.

Recently, some stronger versions have also been studied, where the problem is to delete
at most $k$ vertices to get a graph such that every connected component of the resulting graph is at most $\ell$ edges away from being a graph in a graph class $\FF$ (see~\cite{RaiS18,PhilipRS18,0001R16}).
Some examples of $\FF$ that have been studied in this stronger version include {\it forest, pseudo-forest or bipartite}. 


\noindent
{\bf Our results:}
In this paper, we address the complexity of a very natural variation of the graph deletion problem, where in the resulting graph, each connected component belongs to one of the finitely many graph classes.
For example, we may want the connected components of the resulting graph to be a clique or a biclique (a complete bipartite graph). It is known that cliques forbid exactly $P_3$s, the induced paths of 
length $2$, and 
bicliques forbid $P_4$ and triangles.  So if we just want every connected component to be a clique or every connected component to be a biclique, then one can find appropriate constant sized subgraphs in the given graph and branch on them (as one would in a hitting set instance). However, if we want each connected component to be a clique or a biclique, such a simple approach by branching over $P_3$, $P_4$, or $K_3$ would not work. Notice that triangles are allowed to be present in clique components and $P_3$s are allowed to be present in biclique components. It is not even clear that there will be a finite forbidden set for this resulting graph class.

Let us formally define the deletion problem below where we want every connected component of the resulting graph to belong to at least one of the graph classes $\Pi_i$ with $i \in [d]$ for some finite integer $d$.

\defparproblem{{\gpionetopiddeletion}}{An undirected graph $G = (V, E)$, an integer $k$, and $d$ graph classes $\Pi_1,\ldots,\Pi_d$.}{$k$}{Is there a subset $Z \subseteq V(G), |Z| \leq k$ such that every connected component of $G - Z$ is in at least one of the graph classes $\Pi_1,\ldots,\Pi_d$?}

Many computational problems that are NP-hard in general graphs get solvable in polynomial time when the graph is restricted to some graph class $\Pi$. If a graph $G$ is such that each connected component belongs to at least one of the graph classes $\Pi_i$ for $i \in [d]$ where a problem is solvable in polynomial time, then in most cases, the problem is solvable in the entire graph $G$ as well. Vertex Deletion problems can be viewed as detecting a few outliers of a graph $G$ so that the graph after removing such outliers belongs to a graph class $\Pi$ where problems are efficiently solvable. Since problems get tractable even when the graph is such that its components belong to efficiently tractable graph classes, the vertex deletion problem corresponding to such scattered graph class is interesting as well.



We look at the case when each problem {\sc $\Pi_i$ Vertex Deletion} is known to be FPT and the property that ``graph $G$ belongs to $\Pi_i$" is expressible in CMSO logic (See Section \ref{sec:prelims} for formal definitions). We call this problem {\indpionetopiddeletion}. We show that this problem is fixed parameter tractable.

\begin{theorem} \label{theorem:ind-trac-general-fpt}
{\indpionetopiddeletion} is FPT with respect to solution size $k$.
\end{theorem}

The problem {\indpionetopiddeletion} covers a wide variety of collections of popular graph classes. Unfortunately, the running time of the algorithm from Theorem \ref{theorem:ind-trac-general-fpt} has gargantuan constant overheads. Hence we look at the special case of {\indpionetopiddeletion} named {\pionetopiddeletion} where each of the graph classes is characterized by a finite forbidden set. 
We get a faster FPT algorithm for {\pionetopiddeletion} using the well-known techniques in parameterized complexity -- iterative compression and important separators.

\begin{theorem} \label{theorem:main-result}
{\pionetopiddeletion} can be solved in time $2^{poly(k)}n^{\OO(1)}$.
\end{theorem}

Here, $poly(k)$ denotes a polynomial in $k$. 

\noindent
{\bf Previous Work:}
While there has been a lot of work on graph deletion and modification problems, one work that comes close to ours is the work by Ganian, Ramanujan and Szeider~\cite{GanianRS17} where they consider the parameterized complexity of finding strong backdoors to a scattered class of CSP instances.  In fact, in their conclusion, they remark that

`Graph modification problems and in particular the study of efficiently
computable modulators to various graph classes has been an integral part of parameterized complexity and
has led to the development of several powerful tools and techniques. We believe that the study of modulators
to ‘scattered graph classes’ could prove equally fruitful and, as our techniques are mostly graph based, our
results as well as techniques could provide a useful starting point towards future research in this direction'.

Our work is a starting point in addressing the parameterized complexity of the problem they suggest.
%
%
%

\noindent
{\bf Our Techniques:}
The FPT algorithm for \indpionetopiddeletion\ in Theorem \ref{theorem:ind-trac-general-fpt} relies on a result by Lokshtanov et al. \cite{lokshtanov2018reducing} that allows one to obtain a (non-uniform) FPT algorithm for CMSO-expressible graph problems by designing an FPT algorithm for the problem on a well-connected class of graphs called unbreakable graphs. For the latter, using the observation that only one connected component after deleting the solution is large, which belongs to some particular class $\Pi_i$, we use the FPT algorithm for {\sc $\Pi_i$-Vertex Deletion} to obtain a modulator to the graph class $\Pi_i$ of size $s(k)$ for a function $s$. Then we use a branching rule to remove the components that are not in $\Pi_i$ in the modulator ``revealing" the solution to the problem. See Section \ref{section:general-graphs} for more details.

We now give a brief summary of the FPT algorithm of Theorem \ref{theorem:main-result} for the problem {\pionetopiddeletion}. 
Using the standard technique of iterative compression, we can assume that we have a solution $W$ of size $k+1$ for our problem. The problem can be divided into two cases depending on whether $W$ gets disconnected by the solution or not. If it does not, a simple algorithm via branching on vertices of the finite forbidden graphs plus some important separators \cite{marx2006parameterized} of the graph solves the problem. Else, we can assume that the solution contains a ``special" important separator. In this case, we come up with a recursive procedure to find a set $\mathcal{R}$ of $2^{poly(k)}$ vertices, at least one of which hits the solution. Hence we devise a branching rule on $\mathcal{R}$ solving the problem. See Section \ref{section:finite-forbidden-classes} for more details.

In the procedure to obtain $\mathcal{R}$, the graph in the recursive procedure is obtained by gluing a graph of $poly(k)$ vertices to a subgraph of $G$ along with a set of `boundary' vertices. The techniques we use here is very similar to the one used by Ganian et al. \cite{GanianRS17} where they studied a similar problem to identify outlier variable to a collection of easy Constraint Satisfaction Problems (CSPs). Similar such techniques involving tight separator sequences are used to give FPT algorithms for problems such as {\sc Parity Multiway Cut} \cite{lokshtanov2012parameterized}, {\sc Directed Feedback Vertex Set} \cite{lokshtanov2016linear}, {\sc Subset Odd Cycle Transversal} \cite{lokshtanov2017hitting} and {\sc Saving Critical Nodes with Firefighters} \cite{choudhari2017saving}.

There is a crucial distinction of our algorithm from the problems listed above solved via similar techniques. In the latter, adding and removing edges and vertices are possible. This is used to create gadgets that preserve some properties in the recursive input graph such as connectivity and parity of paths between pairs of vertices. In {\pionetopiddeletion}, we are not allowed to add or remove edges or vertices as doing so might create or destroy forbidden graphs corresponding to graph classes $\Pi_i$ for $i \in [d]$ changing the problem instance. We circumvent this difficulty by just doing contractions of degree two paths in the graph up to a certain constant. 

\section{Preliminaries}
\label{sec:prelims}

{\bf Graph Theory:}
For $\ell \in \mathbb{N}$, we use $P_{\ell}$ to denote the path on $\ell$ vertices. We use standard graph theoretic terminology from Diestel's book~\cite{diestel-book}.
A {\em tree} is a connected graph with no cycles.
A {\em forest} is a graph, every connected component of which is a tree.
A {\em paw} is a graph $G$ with vertex set $V(G) = \{x_1, x_2, x_3, x_4\}$ and edge set $E(G) = \{x_1 x_2, x_2 x_3, x_3 x_1, x_3 x_4\}$.
A graph is a {\em block graph} if all its biconnected components are cliques.
For a set $X \subseteq G$, we use $G[X]$ to denote the graph induced on the vertex set $X$ and we use $G - X$ (or $G \setminus X)$ to denote the graph induced by the vertex set $V(G) \setminus X$.
We say that a subset $Z \subseteq V(G)$ {\em disconnects} a subset $S \subseteq V(G)$ if there exists $v,w \in S$ with $v\neq w$ such that $v$ and $w$ occur in different connected components of the graph $G \setminus Z$. We call a path in a graph $P$ as a degree $2$ path if all the internal vertices of the path have degree $2$ in the  graph $G$.

\begin{definition} Let $G$ be a graph and disjoint subsets $X,S \subseteq V(G)$. We denote by $R_G(X,S)$ the set of vertices that lie in the connected component containing $X$ in the graph $G \setminus S$. We denote $R_G[X,S] = R_G(X,S) \cup S$. Finally we denote $NR_G(X,S) = V(G) \setminus R_G[X,S]$ and $NR_G[X,S] = NR_G(X,S) \cup S$. We drop the subscript $G$ if it is clear from the context.
\end{definition}

\begin{definition} \cite{marx2006parameterized} Let $G$ be a graph and $X,Y \subseteq V(G)$.
\begin{itemize}
\item A vertex set $S$ disjoint from $X$ and $Y$ is said to disconnect $X$ and $Y$ if $R_G(X,S) \cap Y = \phi$. We say that $S$ is an $X-Y$ \textbf{separator} in the graph $G$.
\item  An $X-Y$ separator is \textbf{minimal} if none of its proper subsets is an $X-Y$ separator.
\item  An $X-Y$ separator $S_1$ is said to \textbf{cover} an $X-Y$ separator $S$ with respect to $X$ if $R(X,S) \subset R(X,S_1)$.
\item Two $X-Y$ separators $S_1$ and $S_2$ are said to be incomparable if neither covers the other.
\item In a set $\mathcal{H}$ of $X-Y$ separators, a separator $S$ is said to be component-maximal if there is no separator $S'$ in $\mathcal{H}$ which covers $S$. Component-minimality is defined analogously.
\item An $X-Y$ separator $S_1$ is said to dominate an $X-Y$ separator $S$ with respect to $X$ if $|S_1| \leq |S|$ and $S_1$ covers $S$ with respect to $X$.
\item We say that $S$ is an important $X-Y$ separator if it is minimal and there is no $X-Y$ separator dominating $S$ with respect to $X$.
\end{itemize}
\end{definition}

For the basic definitions of Parameterized Complexity, we refer to \cite{cygan2015parameterized}.

\noindent
{\bf Parameterized Complexity:}
A parameterized problem $L$ is a subset of 
$\Sigma^* \times \mathbb{N}$ for some finite alphabet $\Sigma$.
An instance of a parameterized problem is denoted by $(x, k)$ where $x \in \Sigma^*, k \in \mathbb{N}$.
We assume that $k$ is given in unary and without loss of generality $k \leq |x|$.

\begin{definition}[Fixed-Parameter Tractability]
\label{defn:fpt}
A parameterized problem $L \subseteq \Sigma^* \times \mathbb{N}$ is said to be {\em fixed-parameter tractable} (FPT) if there exists an algorithm $\AAA$, a computable function $f: \mathbb{N} \rightarrow \mathbb{N}$ and a constant $c$ independent of $f, k, |x|$, such that given input $(x, k)$, runs in time $f(k)|x|^{c}$ and correctly decides whether $(x, k) \in L$ or not.
\end{definition}

%
See Cygan et al.~\cite{CFKLMPPS15} for more details on Parameterized Complexity.

\noindent
\textbf{Counting Monadic Second Order Logic.} The syntax of Monadic Second Order Logic
(MSO) of graphs includes the logical connectives $\lor, \land, \lnot, \leftrightarrow , \implies$, variables for vertices, edges, sets of vertices and sets of edges, the quantifiers $\forall$ and $\exists$, which can be applied to these variables, and five binary relations:
\begin{enumerate}
\item $u \in U$, where $u$ is a vertex variable and $U$ is a vertex set variable;
\item $d \in D$, where $d$ is an edge variable and $D$ is an edge set variable;
\item $inc(d, u)$, where $d$ is an edge variable, $u$ is a vertex variable, and the interpretation is
that the edge $d$ is incident to $u$;
\item $adj(u, v)$, where $u$ and $v$ are vertex variables, and the interpretation is that $u$ and $v$ are
adjacent;
\item equality of variables representing vertices, edges, vertex sets and edge sets.
\end{enumerate}
Counting Monadic Second Order Logic (CMSO) extends MSO by including atomic sentences testing whether the cardinality of a set is equal to $q$ modulo $r$, where $q$ and $r$ are integers such that $0 \leq q < r$ and $r \geq 2$. That is, CMSO is MSO with the following atomic sentence: $card_{q,r}(S) = true$ if and only if $|S| \equiv q \mod r$, where $S$ is a set. We refer to \cite{courcelle1990monadic,arnborg1991easy} for a detailed introduction to CMSO.

\section{FPT Algorithm for {\indpionetopiddeletion}} \label{section:general-graphs}

We first formally define the problem. 

\defparproblem{{\indpionetopiddeletion}}{An undirected graph $G = (V, E)$, an integer $k$, and $d$ hereditary graph classes $\Pi_1,\ldots,\Pi_d$ such that for all $i \in [d]$, {\sc $\Pi_i$ Vertex Deletion} is FPT and properties $P_i(H)$ for input graph $H$ is CMSO expressible.}{$k$}{Is there a subset $Z \subseteq V(G), |Z| \leq k$ such that every connected component of $G - Z$ is in at least one of the graph classes $\Pi_1,\ldots,\Pi_d$?}

We recall the notion of unbreakable graphs from \cite{lokshtanov2018reducing}.
\begin{definition}
\label{defn:unbreakable-graph}
A graph $G$ is {\em $(s, c)$-unbreakable} if there does not exist a partition of the vertex set into three sets $X, C$ and $Y$ such that
(a) $C$ is an $(X, Y)$-separator: there are no edges from $X$ to $Y$ in $G \setminus C$, (b) $C$ is {\em small}: $|C| \leq c$, and (c) $X$ and $Y$ are {\em large}: $|X|, |Y | \geq s$.
\end{definition}

We now use the following theorem from \cite{lokshtanov2018reducing} which says that if the problem is FPT in unbreakable graphs, then the problem is FPT in general graphs. Let $CMSO[\psi]$ denote the problem with graph $G$ as an input, and the objective is to determine whether $G$ satisfies $\psi$.

\begin{theorem}\cite{lokshtanov2018reducing} \label{theorem:cmso-unbreakable}
Let $\psi$ be a CMSO sentence.  For all $c \in \mathbb{N}$, there exists $s \in \mathbb{N}$ such that if there exists an algorithm that solves $CMSO[\psi]$ on $(s, c)$-unbreakable graphs in time $\OO(n^d)$ for some $d > 4$, then there exists an algorithm that solves $CMSO[\psi]$ on general graphs in time $O(n^d)$.
\end{theorem}

We prove the following lemma.

\begin{lemma}\label{lemma:problem-in-cmso}
\indpionetopiddeletion\ is CMSO expressible.
\end{lemma}
\begin{proof}
We use $conn(X)$ which verifies that a subset $X$ of a graph $G$ induces a connected subgraph. It is known that $conn(X)$ is expressible by an MSO formula \cite{cygan2015parameterized}. Also for $X \subseteq V(G)$ , we can express the sentence ``$|X| = k$'' as $\exists x_1, \dotsc, x_k \forall u \in V(G) (u \in X) \implies (\lor_{i \in [k]}u = x_i)$


Recall that $P_i(G)$ denote the graph property ``graph $G$ is in $\Pi_i$" for $i \in [d]$ and input graph $G$. Let the CMSO sentences for properties $P_i(G)$ be $\psi_i(G)$. The overall CMSO sentence for our problem $\psi$ is $\exists X \subseteq V(G), |X| = k,    \forall C \subseteq V(G) \setminus X : conn(C) \implies (\lor_{i \in [d]} \psi_i(G[C]))$.
\end{proof}

Hence Theorem \ref{theorem:cmso-unbreakable} and Lemma \ref{lemma:problem-in-cmso}  allow us to focus on unbreakable graphs.

\begin{theorem}\label{theorem:ind-trac-unbreakable}
\indpionetopiddeletion\ is FPT in $(s(k),k)$-unbreakable graphs for any function $s$ of $k$.
\end{theorem}
\begin{proof}

Let $G$ be an $(s(k),k)$-unbreakable graph and $X$ be a solution of size $k$. Look at the connected components of $G-X$. Since $X$ is a separator of size at most $k$, at most one connected component 
 of $G-X$ has size more than $s(k)$.
 
Let us first look at the case where no connected component of $G-X$ has size more than $s(k)$. In this case, we can bound the number of connected components by $2s(k)$. Suppose not. Then we can divide vertex sets of connected components into two parts $C_1$ and $C_2$, each having at least $s(k)$ vertices. Then the partition $(C_1,X,C_2)$ of $V(G)$ contradicts that $G$ is $(s(k),k)$-unbreakable. 

Since each component has size at most $s(k)$, we have $|V(G) \setminus X| \leq 2(s(k))^2$. Hence $|V(G)|  \leq 2(s(k))^2+k$. We can solve the problem by going over all subsets of size $k$ in $G$ and checking if every connected component of $G-X$ is in some graph class $\Pi_i$ for $i \in [d]$. This gives us an algorithm with  running time ${h(k) \choose k} h(k)^{\OO(1)}$ where $h(k) = 2(s(k))^2+k$.


Let us now look at the case where there is a component $C$ of $G-X$ of size more than $s(k)$.  Let $\Pi_j$ be the graph class which $C$ belongs to. Let $R = V -(X \cup C)$. Since $X$ is a separator of size at most $k$ with separation $(C,R)$, we can conclude that $|R| \leq s(k)$. Hence we can conclude that $X \cup R$ is a modulator of size at most $s(k)+k$ such that $G-(X \cup R)$ is a graph in graph class $\Pi_i$.  Hence we can conclude that $G$ has a modulator of size at most $g(k) = s(k)+k$ to the graph class $\Pi_j$.

Our algorithm first guesses the graph class $\Pi_j$ and then uses the FPT algorithm for $\Pi_j$-Vertex Deletion to find a modulator $S$ of size $g(k)$ such that $G-S$ is in the graph class $\Pi_j$.

We know that $(C,X,R)$ is a partition of $V(G)$. Let $(S_{CX}, S_R)$ be the partition of $S$ where $S_{CX} = S \cap (C \cup X)$ and $S_R = S \cap R$. The algorithm goes over all 2-partitions of $S$ to guess the partition $(S_{CX}, S_R)$.

\begin{claim} For every component $Q$ in the graph $G-X$ such that $Q$ is not in the graph class $\Pi_j$, we have $S_R \cap V(Q) \neq \emptyset$.
\end{claim}
\begin{proof}
Suppose $S_R \cap V(Q) = \emptyset$. By definition, we have $V(Q) \subseteq R$. Hence if $S_R \cap V(Q) = \emptyset$, we have $S \cap V(Q) = \emptyset$. But then this contradicts the fact that $G-S$ is in the graph class $\Pi_j$ as $Q$ is not in the graph class $\Pi_j$ and $\Pi_j$ is a hereditary graph class.
\end{proof}

For every vertex $v \in S_R$, let $Q_v$ denote the component in $G-X$ that contains $v$. Note that the neighborhood of $V(Q_v)$ in the graph $G$ is a subset of $X$ which is of size at most $k$. We now use the following proposition that helps us to guess the subset $V(Q_v)$.
\begin{proposition}(\cite{fomin2012treewidth}) \label{prop-2}
Let $G = (V, E)$ be a graph. For every $v \in V$ , and $b, f \geq 0$, the number of connected vertex subsets $B \subseteq V$ such that 
\begin{itemize}
\item $v \in B$
\item $|B| = b+1$
\item $|N(B)| = f$
\end{itemize}
is at most ${b+f \choose b}$ and can be enumerated in time $O(n {b+f \choose b})$ by making use of polynomial space.
\end{proposition}

We have the following Branching Rule.

\begin{branching rule}\label{branching-rule-1}
 Let $v \in S_R$. Using the enumeration algorithm from Proposition \ref{prop-2}, go over all connected vertex subsets $B \subseteq V$ such that $ v \in B, |B| = b + 1,$ and $|N(B)| = f$ where $ 1 \leq b \leq s(k)$ and $ 1 \leq f \leq k$ and return the instance $(G-B, k- |N(B)|)$.
\end{branching rule}

The branching rule is safe because in one of the branches, the algorithm rightfully guesses $B = V(Q_v)$. The algorithm repeats the branching rule for all vertices $v \in S_R$. Hence we can assume that the current instance is such that $S_R = \emptyset$. We update the sets $X$ and $R$ by accordingly deleting the removed vertices. Let $(G',k')$ be the resulting instance. We have the following claim.

\begin{claim} The set $X$ is such that $|X| \leq k'$ and $G-X$ is in the graph class $\Pi_j$.
\end{claim}
The proof of the claim comes from the fact as $S \cap R = \emptyset$, every component other than $C$ does not intersect with $S$. Hence these components have to be in the graph class $\Pi_j$ as $G-S$ is in the graph class $\Pi_j$.
%
%

The algorithm now again uses the FPT algorithm for the graph class $\Pi_j$ to obtain the solution of size $k'$ thereby solving the problem. We summarize the algorithm below.

\begin{enumerate}
\item For any of the given graph classes check whether the given graph $G$ has a modulator of size at most $g(k)$. If none of them has, then return NO. Otherwise, let $\Pi_j$ be such a graph class with $S$ being the modulator.
\item  Go over all 2-partitions $(S_{CX},S_R)$. For each $v \in S_R$, apply Branching Rule \ref{branching-rule-1}. Let $(G',k')$ be the resulting instance.
\item Check whether the graph $G'$ has a $\Pi_j$-deletion set of size at most $k'$. If yes, return YES. Else return NO.
\end{enumerate}

Running Time: Let $f_j(k)n^{O(1)}$ be the running time for $\Pi_j$-Vertex Deletion. We use $j \cdot f_j(g(k)) \cdot 2^{g(k)} n^{O(1)}$ time to obtain set $S$ and its 2-partition where $g(k) = s(k)+k$. We use overall $O(n(g(k)+1))^{k+1})$ time to enumerate the connected vertex sets in Branching Rule \ref{branching-rule-1}. The branching factor is bounded by $(g(k)+1))^{k+1}$ and the depth is bounded by $k$. Since choices of $v$ is bounded by $g(k)$, exhaustive application of Branching Rule takes at most $(g(k))^{k(k+2)}n^{O(1)}$ time. Finally we apply the algorithm for $\Pi_j$-Vertex Deletion again taking at most  $f_j(k)n^{O(1)}$ time.

Hence the overall running time is bounded by $j \cdot f_j(g(k)) \cdot 2^{g(k)} (g(k)+1))^{k(k+2)}n^{O(1)}$.
\end{proof}

The proof of Theorem \ref{theorem:ind-trac-general-fpt} follows from from Theorem \ref{theorem:cmso-unbreakable} and Theorem \ref{theorem:ind-trac-unbreakable}.


\section{Deletion to scattered classes with finite forbidden families} \label{section:finite-forbidden-classes}



Unfortunately, the algorithm for \indpionetopiddeletion\ in Theorem \ref{theorem:ind-trac-general-fpt} has a huge running time due to the gargantuan overhead from applying Theorem \ref{theorem:cmso-unbreakable}.
We now look into a special case of \indpionetopiddeletion\ where every graph class $\Pi_i$ with $i \in [d]$ can be characterized by a finite forbidden family. Note that {\sc $\Pi_i$ Vertex Deletion} is FPT for each $i \in [d]$ from the simple branching algorithm over vertices of the induced subgraphs $H$ of the input graph $G$ that is isomorphic to members of the finite forbidden family $\FF_i$. Also, the properties that "graph $G$ is in $\Pi_i$" can be expressed in CMSO logic as we can hard code the graphs in $\FF_i
$ in the formula. Hence the problem is indeed a special case of \indpionetopiddeletion. In this section, we give an algorithm for this case with running time much better when compared to that in Theorem \ref{theorem:ind-trac-general-fpt}.

We have the following definition.
\begin{definition}
We call a set $Z$ a \textbf{\pionetopidmodulator} if every connected component of $G \setminus Z$ is in one of the graph classes $\Pi_i$ for $i \in d$.
\end{definition}

\noindent
{\bf Brief Outline of the section:}

In Section \ref{section:iterative-compression}, we first use the standard technique of iterative compression to obtain a tuple $(G,k,W)$ of the input instance {\disjointpionetopiddeletioncomp} where $W$ is a \pionetopidmodulator\ of size at most $k+1$ and the aim is to obtain a solution of size at most $k$ disjoint from $W$. We also add an additional requirement to the problem that some of the vertices cannot be in the solution which will be useful later.

In Subsection \ref{section:non-separating}, we give an FPT algorithm for \disjointpionetopiddeletioncomp\ in the special case when the solution that we are looking for leaves $W$ in a single component. The algorithm uses the standard technique of important separators \cite{marx2006parameterized}.

Finally in Subsection \ref{section:general-instance}, we handle general instances of {\disjointpionetopiddeletioncomp}. We focus on instances where the solution separates $W$. We guess $W_1 \subset W$ as the part of $W$ that occurs in some single connected component after deleting the solution. The algorithm finds a set $\mathcal{R}$ of $2^{poly(k)}$ vertices one of which intersects the solution and do a branching on vertices of $\mathcal{R}$. Finding $\mathcal{R}$ involves a recursive subprocedure. 

Since, the solution separates $W$, we know that it contains a $W_1 - (W \setminus W_1)$ separator $X$. It can the proven that $X$ is a `special' kind of important separator (whose definition is tailored to the problem).
 The algorithm uses the technique of tight separator sequences \cite{lokshtanov2012parameterized}. 
It guesses the integer $\ell$ which is the size of the part of the solution present in the graph containing $W_1$ after removing $X$. The algorithm then constructs the important separator sequence corresponding to $\ell$ and finds the separator $P$ furthest from $W_1$ in the sequence such that there is a {\pionetopidmodulator} of size $\ell$ in the graph containing $W_1$ after removing $P$. If separator $X$ either intersects $P$ or dominates the other, then a recursive smaller instance is easily constructable. In the case when the two separators are incomparable, the algorithm identifies a set of vertices $Y$ that is reachable from $W_1$ after deleting $P$. The algorithm then constructs a graph gadget of $k^{\OO(1)}$ vertices whose appropriate attachment to the boundary $P$ of the graph $G[Y]$ gives a graph $G'$ which preserves the part of the solution of $G$ present in $G[Y]$. 
 Since this part of the solution is strictly smaller in size, the algorithm can find the set of vertices hitting the solution $\mathcal{R}$ for $G$ by recursively finding a similar set in $G'$.


\subsection{Iterative Compression}\label{section:iterative-compression}

We use the standard technique of iterative compression to transform the \pionetopiddeletion\ problem into the following problem {\sc Disjoint ($\Pi_1, \Pi_2, \dotsc , \Pi_d$) Vertex Deletion Compression}(\disjointpionetopiddeletioncomp) such that an FPT algorithm with running time $\OO^*(f(k))$ for the latter gives a $\OO^*(2^{k+1}f(k))$ time algorithm for the former. 

\defparproblem{\disjointpionetopiddeletioncomp}{A graph $G$, an integer $k$, finite forbidden sets $\FF_1, \FF_2, \dotsc, \FF_d$ for graph classes $\Pi_1, \Pi_2, \dotsc, \Pi_d$ and a subset $W$ of $V(G)$ such that $W$ is a ($\Pi_1, \Pi_2, \dotsc , \Pi_d$)- modulator of size $k+1$.}{$k$}{Is there a subset $Z \subseteq V(G) \setminus W, |Z| \leq k$ such that $Z$ is a \pionetopidmodulator of the graph $G$?}

We now define an extension of \disjointpionetopiddeletioncomp\ to incorporate the notion of undeletable vertices. The input additionally contains a set $U \subseteq V(G)$ of undeletable vertices and we require the solution $Z \subseteq V(G)$ to be disjoint from $U$.

\defparproblem{{\disjointpionetopiddeletioncompu}}{A graph $G$, an integer $k$, finite forbidden sets $\FF_1, \FF_2, \dotsc, \FF_d$ for graph classes $\Pi_1, \Pi_2, \dotsc, \Pi_d$ a subset $W$ of $V(G)$ such that $W$ is a ($\Pi_1, \Pi_2, \dotsc , \Pi_d$)- modulator of size $k+1$ and a subset $U \subseteq V(G)$.}{$k$}{Is there a subset $Z \subseteq V(G) \setminus (W \cup U), |Z| \leq k$ such that $Z$ is a \pionetopidmodulator of the graph $G$?}

We have the following reduction rule.
\begin{reduction rule}\label{rr1}
If a connected component of $G$ belongs to some graph class $\Pi_i$, then remove all the vertices of this connected component.
\end{reduction rule}

\begin{lemma}\label{lemma:rr1-safety}
 Reduction Rule \ref{rr1} is safe.
\end{lemma}

\begin{proof}
Let $\mathcal{X}$ be the connected component of $G$ removed to get an instance $(G',k,W')$. We claim that $(G,k,W)$ is a YES-instance if and only if $(G',k,W')$ is also a YES-instance. Let $Z$ be a solution of $G$  of size at most $k$. Since $G'$ is an induced subgraph of $G$, $Z$ is also a solution of $G'$ as well. Conversely, suppose $Z'$ is the solution of size $k$ for graph $G'$. Then every connected component of the graph $G' \setminus Z'$ belongs to some graph class $\Pi_i$ for $i \in [d]$. Since $\mathcal{X}$ also belongs to some graph class $\Pi_i$ for some $i \in [d]$, we have that $Z'$ is also a solution for the graph $G$.
\end{proof}

We now develop the following notion of forbidden sets which can be used to identify if a connected component of a graph belongs to any of the classes $\Pi_i$ for $i \in [d]$.

\begin{definition}
We say that a subset of vertices $C \subseteq V(G)$ is a \textbf{forbidden set} 
 of $G$ if $C$ occurs in a connected component of $G$ and
 there exists a subset $C_i \subseteq C$ such that $G[C_i] \in \FF_i$ for all $i \in [d]$ and $C$ is a minimal such set. 
\end{definition}

Clearly, if a connected component of $G$ contains a forbidden set, then it does not belong to any of the graph classes $\Pi_i$ for $i \in [d]$. We note that even though the forbidden set $C$ is of finite size,
 the lemma below rules out the possibility of a simple algorithm involving just branching over all the vertices of $C$.

\begin{lemma}\label{lemma:hit_or_disconnect}
Let $G$ be a graph and $C \subseteq V(G)$ be a forbidden set of $G$. Let $Z$ be a \pionetopidmodulator of $G$. Then $Z$ disconnects $C$ or $Z \cap C \neq \emptyset$.
\end{lemma}

\begin{proof}
Suppose $Z$ is disjoint from $C$. We know that $C$ cannot occur in a connected component $\mathcal{X}$ of $G \setminus Z$ as $\mathcal{X}$ cannot belong to any graph class $\Pi_i$ for $i \in [d]$ due to the presence of subsets $C_i \subseteq C$ such that $G[C_i] \in \FF_i$. Hence $Z$ disconnects $C$.
\end{proof}

\subsection{Finding non-separating solutions}\label{section:non-separating}

\sloppy In this section, we focus on solving instances of {\disjointpionetopiddeletioncompu} which have a non-separating property defined as follows.

\begin{definition}
Let $(G,k, W)$ be an instance of {\disjointpionetopiddeletioncompu} and $Z$ be a solution for this instance. Then $Z$ is called a \textbf{non-separating} solution if $W$ is contained in a single connected component of $G \setminus Z$ and \textbf{separating} otherwise. If an instance has only separating solutions, we call it a separating instance. Otherwise, we call it non-separating.
\end{definition}

We now describe the following lemma on important separators which is helpful in our algorithm to compute non-separating solutions with undeletable vertices. To get separators of size at most $k$ which is disjoint from an undeletable set $U$, we replace each vertex $u \in U$ with $k+1$ copies of $u$ that forms a clique.

\begin{lemma}\label{lemma:imp-sep} \cite{chen2009improved} For every $k \geq 0$ and subsets $X,Y, U \subseteq V(G)$, there are at most $4^k$ important $X-Y$ separators of size at most $k$ disjoint from $U$. Furthermore, there is an algorithm that runs in $\OO(4^k k n)$ time that enumerates all such important $X-Y$ separators and there is an algorithm that runs in $n^{\OO(1)}$ time that outputs one arbitrary component-maximal $X-Y$ separator disjoint from $U$. 
\end{lemma}

We now have the following lemma which connects the notion of important separators with non-separating solutions to our problem. 

\begin{lemma}\label{lemma:non-sep-important-sep}
Let $(G,k,W,U)$ be an instance of \disjointpionetopiddeletioncompu\ obtained after exhaustively applying Reduction Rule~\ref{rr1} and $Z$ be a non-separating solution. Let $v$ be a vertex such that $Z$ is a $\{v\}-W$ separator. Then there is a solution $Z'$ which contains an important $\{v\}-W$ separator of size at most $k$ in $G$ and disjoint from $U$.
\end{lemma}

\begin{proof}
Since we have applied Reduction Rule \ref{rr1} as long as it is applicable, there is no connected component $\mathcal{X}$ of $G$ that is disjoint from $W$. 
 Hence every component of $G$, in particular the component containing $v$ intersects with $W$. Therefore, since the solution $Z$ disconnects $v$ from $W$, it must contain a minimal non-empty $\{v\}-W$ separator $A$ which is disjoint from $U$. If $A$ is an important $\{v\}-W$ separator, we are done. Else there is an important $\{v\}-W$ separator $B$ dominating $A$ which is also disjoint from $U$. We claim that $Z' = (Z \setminus A) \cup B$ is also a solution. Clearly $|Z'| \leq |Z|$. Suppose that there exists a forbidden set $C$ in the graph $G \setminus Z'$. Let $\mathcal{X}$ be the connected component of $G \setminus Z'$ containing $C$. Suppose $\mathcal{X}$ is disjoint from $W$. Then there exists a connected component $\mathcal{Y}$ of $G \setminus W$ containing $\mathcal{X}$, contradicting that $W$ is a {\pionetopidmodulator}. Hence $\mathcal{X}$ must intersect $W$. Since $B \subseteq Z'$ disconnects $v$ from $W$, we can conclude that $\mathcal{X}$ is not contained in $R_{G}(v,B)$ as if so it cannot intersect with $W$. 

By the definition of $Z'$, any component of the graph $G \setminus Z'$ which intersects $Z \setminus Z' = A \setminus B$ has to be contained in the set $R_{G}(v,B)$. Hence the component $\mathcal{X}$ is disjoint from $Z \setminus Z'$. Thus, there exists a component $\mathcal{H}$ of the graph $G \setminus Z$ containing $\mathcal{X}$. But this contradicts that $Z$ is a {\pionetopidmodulator}.
\end{proof}

We use the above lemma along with Lemma \ref{lemma:imp-sep} to obtain our algorithm for non-separating instances. The algorithm finds a minimal forbidden set $C$ in polynomial time which by definition is of bounded size. Then it branches on the set $C$ and also on $\{v\}-W$ important separators of size at most $k$ of $G$ for all $v \in C$.

\begin{lemma}\label{lemma:non-sep}
Let $(G,k,W,U)$ be a non-separating instance of \disjointpionetopiddeletioncompu. Then the problem can be solved in $2^{\OO(k)}n^{\OO(1)}$ time.
\end{lemma}

\begin{proof}
We first apply Reduction Rule \ref{rr1} exhaustively. 
If the graph is empty, we return YES. 
 Else, there is a connected component of $G$ which does not belong to any graph classes $\Pi_i$ for $i \in [d]$. Therefore, there exists a forbidden set $C \subseteq V(G)$ of $G$ present in this connected component. We find $C$ as follows. 
We check for each graph class $\Pi_i$, if a graph in $\FF_i$ exists as an induced subgraph for a particular connected component $\mathcal{X}$ of $G$. If so, we take the union of the vertices of these induced graphs. We then make the set minimal by repeating the process of removing a vertex and seeing if the set remains a forbidden set. 

We branch in $|C \setminus (W \cup U)|$-many ways by going over all the vertices $v \in C \setminus (W \cup U)$ and in each branch, recurse on the instance $(G - v, k-1, W, U)$. Then for all $v \in C$, we branch over 
all important $\{v\}-W$ separators $X$ of size at most $k$ in $G$ disjoint from $U$ and recurse on instances $(G \setminus X, k- |X|, W, U)$. 

We now prove the correctness of the algorithm. Let $Z \subseteq V(G) \setminus (W \cup U)$ be a solution of the instance. From Lemma \ref{lemma:hit_or_disconnect}, we know that a forbidden set $C$ of $G$ is disconnected by $Z$ or $Z \cap C \neq \emptyset$. In the latter case, we know that $Z$ contains a vertex $x \in C \setminus (W \cup U)$ giving us one of the branched instances obtained by adding $x$ into the solution.

Now we are in the case where $C$ is disconnected by $Z$. Since Reduction rule \ref{rr1} is applied exhaustively, the connected component containing $C$ also contains some vertices in $W$. Since $Z$ is a non-separating solution, $W$ goes to exactly one connected component of $G \setminus Z$ and there exists some non-empty part of $C$ that is not in this component. Hence, there exists some vertex $x \in C$ that gets disconnected from $W$ by $Z$. From Lemma \ref{lemma:non-sep-important-sep}, we know that there is also a solution $Z'$ which contains an important $\{x\}-W$ separator of size at most $k$ in $G$ disjoint from $U$. Since we have branched over all such $\{x\}-W$ important separators disjoint from $U$, we have correctly guessed on one such branch.

We now bound the running time. Since $|C| = \OO(d)$, any forbidden set in $G$ can be obtained via brute force in $n^{\OO(d)}$ time. For each $i \in [k]$, we know that there are at most $4^i$ important separators of size $1 \leq i \leq k$ disjoint from $U$ which can be enumerated using Lemma \ref{lemma:imp-sep} in $\OO(4^i \cdot i \cdot n)$ time. For the instance $(G,k,W)$, if we branch on $v \in  C$ , $k$ drops by 1 and if we branch on a $\{v\}-W$ separator of size $i$, $k$ drops by $i$. Hence if $T(k)$ denotes the time taken for the instance $(G,k,W)$, we get the recurrence relation $T(k) = \OO(d) T(k-1) + \sum\limits_{i=1}^{k} 4^i T(k-i)$. Solving the recurrence taking into account that $d$ is a constant, we get that $T(k) = 2^{\OO(k)}n^{\OO(1)}$.

\end{proof}

\subsection{Solving  general instances} \label{section:general-instance}

We now solve general instances of {\disjointpionetopiddeletioncompu} using the algorithm for solving non-separating instances as a subroutine. Hence we focus on solving separating instances of {\disjointpionetopiddeletioncompu}. We guess a subset $W_1 \subset W$ such that for a solution $Z$, $W_1$ is exactly the intersection of $W$ with a connected component of $G \setminus Z$. 
 For $W_2 = W \setminus W_1$, we are looking for a solution $Z$ containing a $W_1 - W_2$ separator. Formally, let $W = W_1 \uplus W_2$ be a set of size $k+1$ which is a ($\Pi_1, \Pi_2, \dotsc , \Pi_d$)-modulator. We look for a set $Z \subseteq V(G) \setminus (W \cup U)$ of size at most $k$ such that $Z$ is a ($\Pi_1, \Pi_2, \dotsc , \Pi_d$)-modulator, $Z$ contains a minimal $(W_1, W_2)$-separator $X$ disjoint from $U$ and $W_1$ occurs in a connected component of $G \setminus Z$.

From here on, we assume that the separating instance $(G,k,W,U)$ of {\disjointpionetopiddeletioncompu} is represented as $(G,k,W_1,W_2,U)$ where $W = W_1 \uplus W_2$. We branch over all partitions of $W$ into $W_1$ and $W_2$ which adds a factor of $2^{k+1}$ to the running time. 

\subsubsection{Disconnected case}

We first focus on the particular case when the input instance is such that $W_1$ and $W_2$ are already disconnected in the graph $G$. We have the following lemma that allows us to focus on finding a non-separating solution in the connected component containing $W_1$ to reduce the problem instance.

\begin{lemma}\label{lemma:S-split}
Let $\II = (G,k,W_1,W_2,U)$ be an instance of {\disjointpionetopiddeletioncompu} where $W_1$ and $W_2$ are in distinct components of $G$. Let $Z$ be its solution such that $W_1$ exactly occurs in a connected component of $G \setminus Z$. Also let $R(W_1)$ be the set of vertices reachable from $W_1$ in $G$. Let $Z'  = Z \cap R(W_1)$. Then $(G[R(W_1)],|Z'|,W_1, U \cap R(W_1))$ is a non-separating YES-instance of \disjointpionetopiddeletioncompu\ and conversely for any non-separating solution $Z''$ for $(G[R(W_1)],|Z'|,W_1, U \cap R(W_1))$, the set $\hat{Z} = (Z \setminus Z') \cup Z''$ is a solution for the original instance such that $W_1$ exactly occurs in a connected component of $G \setminus Z''$.
\end{lemma}
\begin{proof}
Suppose $Z'$ is not a \pionetopidmodulator for the graph $G'$. Then some component of $G' \setminus Z'$ contains a forbidden set $C$. The sets $Z'$ and $Z \setminus Z'$ are disjoint as $W_1$ and $W_2$ are disconnected in $G$. Hence $C$ is also in a connected component of $G \setminus Z$ giving a contradiction. Hence $Z'$ is a solution for the instance $(G',|Z'|,W_1)$. Since the solution $Z$ is such that $W_1$ is contained in a connected component of $Z$ and $Z \setminus Z'$ is disconnected from from $Z'$, $Z'$ is a non-separating solution.

Conversely, suppose $\hat{Z}$ is not a solution for the graph $G$. Then there exists a forbidden $C$ in a connected component of $G \setminus \hat{Z}$. 
 Either $C$ is contained in the set $R(W_1)$ or in the set $NR(W_1) = V(G) \setminus R(W_1)$. If $C \subseteq R(W_1)$, $C$ is also present in a connected component of the graph $G' \setminus Z''$ giving a contradiction that $Z''$ is a \pionetopidmodulator of $G'$. If $C \subseteq NR(W_1)$, then $C$ is contained in some connected component of the graph $G[NR(W_1)] \setminus (Z \setminus Z')$. Since $Z'$ is disjoint from $C$, we conclude that $C$ is a forbidden set in the graph $G \setminus ( Z'\cup (Z \setminus Z')) = G \setminus Z$, giving a contradiction.
\end{proof}

We have the following reduction rule.

\begin{reduction rule}\label{rr2}
Let $\II = (G,k,W_1,W_2,U)$ be an instance of \disjointpionetopiddeletioncompu\ where $W_1$ and $W_2$ are disconnected in $G$.  Compute a non-separating solution $Z'$ for the instance $(G',k',W_1, U')$ where $G' = G[R(W_1)]$, $U' = U \cap R(W_1)$ and $k'$ is the least integer $i \leq k$ for which $(G',i,W_1, U')$ is a YES-instance. Delete $Z'$ and return the instance $(G \setminus Z', k - |Z'|, W_2, U)$.
\end{reduction rule}

The safeness of Reduction Rule \ref{rr2} follows from Lemma \ref{lemma:S-split}. The running time for the reduction is $2^{\OO(k)}n^{\OO(1)}$ which comes from that of the algorithm in Lemma \ref{lemma:non-sep}. \\ 

We now introduce the notion of tight separator sequences and $t$-boundaried graphs which are used to design the algorithm.

\subsubsection{Good Separators and Tight Separator Sequences}

We first look at a type of $W_1 - W_2$ separators such that the graph induced on the vertices reachable from $W_1$ after removing the separator satisfies the property as defined below.

\begin{definition} \sloppy Let $(G,k,W_1,W_2,U)$ be an instance of {\disjointpionetopiddeletioncompu}. For integer $\ell$, we call a $W_1 - W_2$ separator $X$ in $G$ \textbf{($\ell$,$U$)-good} if there exists a set $K$ of size at most $\ell$ such that $K \cup X$ is a {\pionetopidmodulator} for the graph $G[R[W_1,X]]$ with $(K \cup X) \cap U = \emptyset$. Else we call it \textbf{($\ell$,$U$)-bad}. If $U = \emptyset$, we call it $\ell$-good and $\ell$-bad respectively.
\end{definition}

We now show that ($\ell$,$U$)-good separators satisfy a monotone property.

\begin{lemma}\label{lemma:monotonicity}
\sloppy Let $(G,k,W_1,W_2,U)$ be an instance of \disjointpionetopiddeletioncompu\ and let $X$ and $Y$ be disjoint $W_1-W_2$ separators in $G$ such that $X$ covers $Y$ and $(X \cup Y) \cap U = \emptyset$. If the set $X$ is ($\ell$,$U$)-good, then $Y$ is also ($\ell$,$U$)-good. 
\end{lemma}

\begin{proof}
Let us define graphs $G_X = G[R[W_1,X]]$ and $G_Y = G[R[W_1,Y]]$. Let $K$ be a subset of size at most $\ell$ such that $K \cup X$ is a {\pionetopidmodulator} for the graph $G_X$ with $(K \cup X) \cap U = \emptyset$. Let $K' = K \cap R[W_1,Y]$. Note that since $K' \subseteq K$, we have $K' \cap U = \emptyset$. We claim that $K' \cup Y$ is a \pionetopidmodulator\ for the graph $G_Y$ proving that $Y$ is ($\ell$,$U$)-good.

Suppose $K' \cup Y$ is not a \pionetopidmodulator. Then there exists a forbidden set $C \subseteq R[W_1,Y]$ contained in a single component of $G_Y \setminus (K' \cup Y)$. Since $C \subseteq R[W_1,Y] \subset R[W_1,X]$ and $X$ and $Y$ are disjoint, $C$ does not intersect $X$. Also $C$ does not contain any vertices in $K \setminus K'$ as $Y$ disconnects the set from $C$. Hence $C$ is disjoint from $K \cup X$.   Since $C$ lies in a single connected component of $G_Y \setminus (K' \cup Y)$ 
 we can conclude that $C$ occurs in a  single connected component of the graph $G_X \setminus (K \cup X)$ giving a contradiction that $X$ is ($\ell$,$U$)-good.
\end{proof}

\begin{definition} \sloppy Let $(G,k,W_1,W_2,U)$ be an instance of \disjointpionetopiddeletioncompu\ and let $X$ and $Y$ be $W_1 - W_2$ separators in $G$ such that $Y$ dominates $X$ and $(X \cup Y) \cap U = \emptyset$. Let $\ell$ be the smallest integer $i$ for which $X$ is ($i$,$U$)-good. If $Y$ is ($\ell$,$U$)-good, then we say that $Y$ \textbf{well-dominates} $X$. If $X$ is ($\ell$,$U$)-good and there is no $Y \neq X$ which well-dominates $X$, then we call $X$ as \textbf{($\ell$,$U$)-important}.
\end{definition}

The following lemma allows us to assume that the solution of the instance $(G,k,W_1,W_2,U)$ contains an ($\ell$,$U$)-important $W_1 - W_2$ separator for some appropriate value of $\ell$.

\begin{lemma}\label{lemma:l-imp-sep}
\sloppy Let $(G,k,W_1,W_2,U)$ be an instance of \disjointpionetopiddeletioncompu\ and $Z$ be a solution.  Let $P \subseteq Z$ be a non-empty minimal $W_1 - W_2$ separator in $G$ and let $P'$ be a $W_1 - W_2$ separator in $G$ well-dominating $P$. Then there is also a solution $Z'$ for the instance containing $P'$.
\end{lemma}

\begin{proof}
Let $Q = Z \cap R[W_1,P]$. Note that $Q$ is a ($\Pi_1, \Pi_2, \dotsc , \Pi_d$)-modulator for the graph $G[R[W_1,P]]$ with $Q \cap U = \emptyset$. Let $Q' \supseteq P'$ be a smallest ($\Pi_1, \Pi_2, \dotsc , \Pi_d$)- modulator for the graph  $G[R[W_1,P']]$ extending $P'$ with $Q' \cap U = \emptyset$. We claim that $Z' = (Z \setminus Q) \cup Q'$ is a solution for the instance $(G,k,W_1,W_2,U)$. Since $P'$ well-dominates $P$, $|Z'| \leq |Z|$ and $Z' \cap U = \emptyset$. Also note that $Z' \cap U = \emptyset$. We now show that $Z'$ is a ($\Pi_1, \Pi_2, \dotsc , \Pi_d$)-modulator. Suppose not. Then there exists a forbidden subset $C$ present in a connected component $\mathcal{X}$ of $G \setminus Z'$.

We first consider the case when $\mathcal{X}$ is disjoint from the set $Z \setminus Z'$. Then there is a component $\mathcal{H}$ in  $G \setminus Z$ which contains $\mathcal{X}$ and hence $C$, contradicting that $Z$ is a solution. We now consider the case when $\mathcal{X}$ intersects $Z \setminus Z'$. By definition of $Z'$, $\mathcal{X}$ is contained in the set $R(W_1,P')$. Since $Z' \setminus Q'$ is disjoint from  $R(W_1,P')$ and is separated from $R(W_1,P')$ by just $P'$, we can conclude that $\mathcal{X}$ and hence $C$ is contained in a single connected component of $G[R[W_1,P']] \setminus Q'$. But this contradicts that $Q'$ is a ($\Pi_1, \Pi_2, \dotsc , \Pi_d$)-modulator in the graph $G[R[W_1,P']]$.
\end{proof}

We now define the notion of a tight separator sequence. It gives a natural way to partition the graph into parts with small \textit{boundaries}. 

\begin{definition} An $X-Y$ \textbf{tight separator sequence of order $k$ with undeletable set $U$} of a graph $G$ with $X, Y, U \subseteq V(G)$ is a set $\mathcal{H}$ of $X-Y$ separators such that 
\begin{itemize}
\item every separator has size at most $k$, 
\item the separators are pairwise disjoint, 
\item every separator is disjoint from $U$, 
\item for any pair of separators in the set, one covers another and 
\item the set is maximal with respect to the above properties.
\end{itemize}
\end{definition}

See figure \ref{fig:separator-sequence} for an example of a tight separator sequence.

\begin{lemma}\label{lemma:tight-separator-sequence}
 Given a graph $G$, disjoint vertex sets $X,Y$ and integer $k$, a tight separator sequence $\mathcal{H}$ of order $k$ with undeletable set $U$ can be computed in $|V(G)|^{\OO(1)}$ time.
\end{lemma}

\begin{proof}
\begin{figure}[t]
\centering
	\includegraphics[scale=0.3]{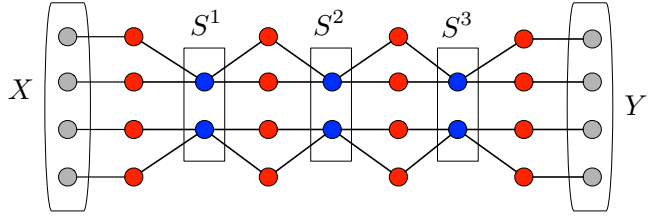}
	\caption{An $X-Y$ tight separator sequence of order two and $U= \emptyset$}
\label{fig:separator-sequence}
\end{figure}

We first replace every vertex $u \in U$ in our graph $G$ with $k+1$ copies of $u$ forming a clique. Note that any $X-Y$ separator of size at most $k$ in the new graph must be disjoint from the vertices of the clique corresponding to every $u \in U$.

In this new graph, we check using the minimum cut algorithm if there is an $X-Y$ separator of size at most $k$. If not, we stop the procedure. Else we compute an arbitrary component-maximal $X-Y$ separator $S$ of size at most $k$ using the polynomial time algorithm in Lemma \ref{lemma:imp-sep}. We add $S$ to the family $\mathcal{H}$, set $Y$ to $S$, and repeat the process. We claim that $\mathcal{H}$ is a tight separator sequence of order $k$ with an undeletable set $U$ after the procedure terminates. It is 
clear that the first four properties of tight separator sequence are satisfied by $\mathcal{H}$ in any iteration. Suppose $\mathcal{H}$ is not maximal and hence an $X-Y$ separator $P$ disjoint from $U$ can be added. If $P$ covers one of the separators $S'$ in $\mathcal{H}$, it contradicts the component-maximality of $S'$ at the time it was added to $\mathcal{H}$. Else $P$ is covered by all the separators in $\mathcal{H}$ which contradicts the termination of the procedure after the last separator in $\mathcal{H}$ was added. This completes the proof.
\end{proof}

In the proof, it can be seen that the separators $S$ in $\mathcal{H}$ can be totally ordered by the subset relation of the reachability sets $R(X, S)$. Hence $\mathcal{H}$ is rather called a sequence than a family of separators.


\subsubsection{Boundaried graphs}

\begin{definition} A \textbf{$t$-boundaried graph} $G$ is a graph with $t$ distinguished labelled vertices. We call the set of labelled vertices $\partial(G)$ the boundary of $G$ and the vertices in $\partial(G)$ terminals. Let $G_1$ and $G_2$ be two $t$-boundaried graphs with the graphs $G_1[\partial(G_1)]$ and $G_2[\partial(G_2)]$ being isomorphic. Let $\mu : \partial(G_1) \rightarrow \partial(G_2)$ be a bijection which is an isomorphism of the graphs $G_1[\partial(G_1)]$ and $G_2[\partial(G_2)]$. We denote the graph $G_1 \otimes_{\mu} G_2$ as a $t$-boundaried graph obtained by the following gluing operation. We take the union of graphs $G_1$ and $G_2$ and identify each vertex $x \in \partial(G_1)$ with vertex $\mu(x) \in \partial(G_2)$. 
 The $t$-boundary of the new graph is the set of vertices obtained by unifying. 
\end{definition}

\begin{definition}
A \textbf{$t$-boundaried graph with an annotated set} is a $t$-boundaried graph with a second set of distinguished but unlabelled vertices disjoint from the boundary. The set of annotated vertices is denoted by $\Delta(G)$.
\end{definition}

\subsubsection{Algorithm}


We design a recursive algorithm {\sc Main-Algorithm} which takes as input the instance $\II = (G,k,W_1,W_2,U)$ and outputs YES if there exists a solution $Z \subseteq V \setminus (W_1 \cup W_2 \cup U)$ such that every connected component of $G-Z$ belongs to some graph class $\Pi_i$ for $i \in [d]$. 

\noindent{\textbf{Description of {\sc Main-Algorithm} procedure:}}
The {\sc Main-Algorithm} procedure initially checks if Reduction Rule \ref{rr1} is applicable for $\II$. Then it checks if $(G,k,W_1 \cup W_2, U)$ is a non-separating YES-instance using the algorithm from Lemma \ref{lemma:non-sep}. 
 If not, it checks if Reduction Rule \ref{rr2} is applicable.

After these steps, we know that any solution $Z$ of $\II$ contains an $(\ell, U)$-good $W_1-W_2$ separator $X$ in the graph $G$ for some integer $0 \leq \ell \leq k$ with $|X| = \lambda > 0$. 
Using Lemma \ref{lemma:l-imp-sep}, we can further assume that the separator $X$ is $(\ell, U)$-important. 
 Since $\lambda > 0$, we have $Z \cap R(W_1,X) \subset Z$ as $X$ is not part of the set $Z \cap R(W_1,X)$. Hence $\ell = |Z \cap R(W_1,X)| < |Z| \leq k$.
Hence we can conclude that $0 \leq \ell < k$ and $1 \leq \lambda \leq k$.

The {\sc Main-Algorithm} procedure 
now calls a subroutine {\sc Branching-Set} with input as $(\II,\lambda,\ell)$ for all values $0 \leq \ell < k$ and $1 \leq \lambda \leq k$.
 The {\sc Branching-Set} subroutine returns a vertex subset $\mathcal{R} \subseteq V(G)$ of size  $2^{poly(k)}$ such that for every solution $Z \subseteq (V(G) \setminus U)$ of the given instance $\II$ containing an $(\ell, U)$-important $W_1-W_2$ separator $X$ of size at most $\lambda$ in $G$, the set $\mathcal{R}$ intersects $Z$. The {\sc Main-Algorithm} procedure then branches over all vertices $v \in \mathcal{R}$ and recursively run on the input $\II' = (G-v,k-1,W_1,W_2, U)$. \\

\noindent{\textbf{Description of {\sc Branching-Set} procedure:}}

We first check if there is a $W_1-W_2$ separator of size $\lambda$ in the graph $G$ with the vertices contained in the set $V \setminus U$. 
 If there is no such separator, we declare the tuple invalid. Else we execute the algorithm in Lemma \ref{lemma:tight-separator-sequence} to obtain a tight $W_1-W_2$ separator sequence $\mathcal{T}$ of order $\lambda$ and undeletable set $U$. 

Let $\mathcal{T}= O_1, O_2, \dotsc, O_q$ for some integer $q$. We partition $\mathcal{T}$ into $(\ell,U)$-good and $(\ell,U)$-bad separators as follows. Recall Lemma \ref{lemma:monotonicity} where we proved that if $X$ and $Y$ are disjoint $W_1-W_2$ separators in $G$ such that $X$ covers $Y$ and $X$ is $(\ell,U)$-good, then $Y$ is also $(\ell,U)$-good. From this we can conclude that the separators in the sequence $\mathcal{T}$ are such that if they are neither all $(\ell,U)$-good nor all $(\ell,U)$-bad,  there exist an $i \in [q]$ where $O_1, \dotsc , O_i$ are $(\ell,U)$-good and $O_{i+1}, \dotsc , O_q$ are $(\ell,U)$-bad. We can find $i$ in $\lceil \log q \rceil$ steps via binary search if at each step, we know of a way to check if for a given integer $j \in [q-1]$ if $O_j$ is $(\ell,U)$-good and $O_{j+1}$ is $(\ell,U)$-bad. In the case $j=q$, we only check if $O_j$ is $(\ell,U)$-good and if so conclude that all the separators in the sequence are $(\ell,U)$-good. In the case where $j=0$, we only check if $O_j$ is $(\ell,U)$-bad and if so conclude that all the separators in the sequence are $(\ell,U)$-bad. 

In any case, we need a procedure to check whether a given separator $P$ is $(\ell,U)$-good or not. From the definition of $(\ell,U)$-good separator, this translates to checking if there is a {\pionetopidmodulator} of size at most $\ell$ in the graph $G[R(W_1,P)]$ such that the solution is disjoint from $W_1 \cup U$. Note that since $P$ separates $W_1$ from $W_2$, $W_1$ is a {\pionetopidmodulator} in the graph $G[R(W_1,P)]$. Hence the problem translates to checking whether $\II_1 = (G[R(W_1,P)], \ell, W_1, U)$ is a YES-instance of {\disjointpionetopiddeletioncompu}.  
This can be done by calling the {\sc Main-Algorithm} procedure for the instance $(G[R(W_1,P)], \ell, W_1, U)$.   
Note that this is a recursive call in the initial {\sc Main-Algorithm} procedure with $\II$ as input where we called the {\sc Branching-Set} procedure with $\ell$ being strictly less than $k$. 

If we do not find an integer $i$ such that $O_j$ is $(\ell,U)$-good and $O_{j+1}$ is $(\ell,U)$-bad, or conclude that all the separators in the sequence are either $(\ell,U)$-good or all are $(\ell,U)$-bad, we declare that the tuple is not valid. Otherwise, we have a separator $P_1$ which is component maximal among all the good separators in $\mathcal{T}$ if any exists, and separator $P_2$ which is component minimal among all the bad separators in $\mathcal{T}$ if any exists. We initialize the set $\mathcal{R}:= P_1 \cup P_2$. For $ i \in \{1,2\}$, we do the following.

We go over every subset $P_i^r \subseteq P_i$. For each such subset, we compute a family $\mathcal{H}$ of $|P_1^r|$-boundaried graphs which consists of all graphs of size at most $k^{5(pd)^2}$ of which at most $k$ are annotated. Note that the total number of such graphs is bounded by $2^{{k^{5(pd)^2} \choose 2}} {k^{5(pd)^2} \choose k+1}$ and these can be enumerated in time  $2^{{k^{5(pd)^2} \choose 2}} {k^{5(pd)^2} \choose k+1} k^{\OO(1)}$. 

For every choice of $P_i^r \subseteq P_i$, for every annotated boundaried graph $\hat{G} \in \mathcal{H}$ with $|P_i^r|$ terminals and every possible bijection $\delta : \partial(\hat{G}) \rightarrow P_i^r$, we construct the glued graph $G_{P_i^r, \delta} = G[R[W_1,P_i]] \otimes_\delta \hat{G}$, where the boundary of $G[R[W_1,P_i]]$ is $P_i^r$. We then recursively call {\sc Branching-Set}$((G_{P_i^r, \delta} \setminus \tilde{S}, k-j, W_1, P_i \setminus P_i^r,  U \cup V(\hat{G}) \setminus P_i^r),\lambda',\ell')$ for every $0 \leq \lambda' < \lambda$, $1 \leq j \leq k-1$ and $0 \leq \ell' \leq \ell$, where $\tilde{S}$ is the set of annotated vertices in $\hat{G}$. We add the union of all the vertices returned by these recursive instances to $\mathcal{R}$ and return the resulting set.

This completes the description of the {\sc Branching-Set} procedure. We now proceed to the proof of correctness. \\

\noindent{\textbf{Correctness of {\sc Main-Algorithm} and {\sc Branching-Set} procedure:}}
We prove the correctness of {\sc Main-Algorithm} by induction on $k$. The case when $k=0$ is correct as we can check if $\II$ is a YES-instance in polynomial time by checking if every connected component of $G$ belongs to one of the graph classes $\Pi_i$ for $i \in [d]$. We now move to the induction step with the induction hypothesis being that the {\sc Main-Algorithm} procedure correctly runs for all instances $\II$ where $k < \hat{k}$ for some $\hat{k} \geq 1$ and identifies whether $\II$ is a YES-instance. We now look at the case when the algorithm runs on an instance with $k = \hat{k}$.


 The correctness of the initial phase follows from the safeness of Reduction Rules \ref{rr1}, \ref{rr2} and the correctness of the algorithm in the non-separating case. Let us now assume that the {\sc Branching-Set} procedure is correct.
 Hence the set $\mathcal{R}$ returned by {\sc Branching-Set} procedure is such that it intersects a solution $Z$ if it exists. Therefore $\II = (G,k,W_1,W_2,U)$ is a YES-instance if and only if $\II' = (G-v,k-1,W_1,W_2,U)$ is a YES-instance for some $v \in \mathcal{R}$. Applying the induction hypothesis for {\sc Main-Algorithm} with input instance $\II'$, we prove the correctness of {\sc Main-Algorithm}.

It remains to prove the correctness of the {\sc Branching-Set} procedure. Note that all the calls of {\sc Main-Algorithm} in this procedure have input instances checking for solutions strictly less than $k$. Hence these calls run correctly from the induction hypothesis when {\sc Branching-Set} is called in the {\sc Main-Algorithm} procedure. Hence we only need to prove that {\sc Branching-Set} procedure is correct with the assumption that all the calls of {\sc Main-Algorithm} in the procedure run correctly.
We prove this by induction on $\lambda$. Recall that the sets $P_1$ and $P_2$ were identified via a binary search procedure described earlier using the calls of {\sc Main-Algorithm} with values strictly less than $k$. Since we assume that the calls of {\sc Main-Algorithm} runs correctly, the sets $P_1$ and $P_2$ were correctly identified if present.

 We first consider the base case when $\lambda = 1$ where there is a $W_1- W_2$ $(\ell,U)$-good separator $X \subseteq Z$ of size one. Since $X$ has size one, it cannot be incomparable with the separator $P_1$. Hence the only possibilities are $X$ is equal to $P_1$, is covered by $P_1$ or covers $P_1$. In the first case, we are correct as $P_1$ is contained in $\mathcal{R}$. The second case contradicts that $X$ is $(\ell,U)$-important $W_1- W_2$ separator. 
 We note that in the third case, we can conclude that $X$ is covered by $P_2$. This is because the other cases where $X$ is equal to be $P_2$ or $X$ covers $P_2$ cannot happen as $P_2$ is $(\ell, U)$-bad and $X$ is incomparable to $P_2$ cannot happen as both are of size one. Hence $X$ covers $P_1$ and is covered by $P_2$. But then $X$ must be contained in the tight separator sequence $\mathcal{T}$ contradicting that $P_1$ is component maximal. Hence the third case cannot happen.

We now move to the induction step with the induction hypothesis being that {\sc Branching-Set} procedure correctly runs for all tuples where $\lambda < \hat{\lambda}$ for some $\hat{\lambda} \geq 2$ and returns a vertex set that hits any solution for its input instance that contains an $(\ell, U)$-important separator of size $\lambda$ and not containing any vertices from $U$. We now look at the case when the algorithm runs on a tuple with $\lambda = \hat{\lambda}$.

Let $Z \subseteq (V(G) \setminus U)$ be a solution for the instance $\II$ containing an $(\ell,U)$-important separator $X$. If $X$ intersects $P_1 \cup P_2$ we are done as $\mathcal{R} \supseteq P_1 \cup P_2$ intersects $X$. Hence we assume that $X$ is disjoint from $P_1 \cup P_2$. Suppose $X$ is covered by $P_1$. Then we can conclude that $P_1$ well-dominates $X$ contradicting that $X$ is $(\ell, U)$-important $W_1- W_2$ separator.

By Lemma  \ref{lemma:monotonicity}, since $X$ is $(\ell,U)$-good and $P_2$ is not, $X$ cannot cover $P_2$. Suppose $X$ covers $P_1$ and itself is covered by $P_2$. Then $X$ must be contained in the tight separator sequence $\mathcal{T}$ contradicting that $P_1$ is component maximal. Hence this case also does not happen. \\

\noindent{\textbf{Incomparable Case:}} 

Finally we are left with the case where $X$ is incomparable with $P_1$ or with $P_2$ if $P_1$ does not exist. Without loss of generality, assume $X$ is incomparable with $P_1$. The argument in the case when $P_1$ does not exist follows by simply replacing $P_1$ with $P_2$ in the proof.

Let $K \subseteq Z$ be the \pionetopidmodulator for the graph $G[R[W_1,X]]$ extending $X$, i.e $X \subseteq K$. In other words, $K = Z \cap R[W_1,X]$.  Since $X$ is an $(\ell, U)$-good separator of $G$, we have $|K \setminus X| \leq \ell$. If $P_1 \cap K$ is non-empty, we have that $P_1 \cap Z$ is non-empty. Since $P_1$ is contained in $\mathcal{R}$, the algorithm is correct as $\mathcal{R}$ intersects $Z$.  Hence we can assume that $P_1$ and $K$ are disjoint.

Let $X^r = R(W_1,P_1) \cap X$ and $X^{nr} = X \setminus X^r$. Similarly, define $P_1^r = R(W_1,X) \cap P_1$ and $P_1^{nr} = P_1 \setminus P_1^r$. Since $X$ and $P_1$ are incomparable, the sets $X^r, X^{nr}, P_1^r$ and  $P_1^{nr}$ are all non-empty. Let $K^r = K \cap R[W_1,P_1]$ and $K^{nr} = K \setminus K^r$. Note that $X^r \subseteq K^r$ and $X^{nr} \subseteq K^{nr}$. See Figure \ref{fig:incomparable-case}.

\begin{figure}[t]
\centering
	\includegraphics[scale=0.3]{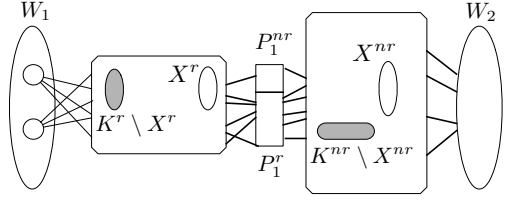}
	\caption{The case where $X$ is incomparable with $P_1$}
\label{fig:incomparable-case}
\end{figure}

We intend to show in the case when $X$ and $P_1$ are incomparable, the set returned by one of the recursive calls of the {\sc Branching-Set} procedure hits the solution $Z$.

We now prove the following crucial lemma where we show that by carefully replacing parts outside of $R[W_1, P_1]$ with a small gadget, we can get a smaller graph $G'$ such that $K^r$, the part of $K$ inside the set $R[W_1, P_1]$ is an optimal \pionetopidmodulator containing the $(|K^r \setminus X^r|, U')$-important separator $X^r$ in this graph for an appropriate subset $U' \subseteq V(G')$. This allows us to show that the instance of {\sc Branching-Set} corresponding to $G'$ would return a set that intersects the optimal \pionetopidmodulator $K^r$ in $G'$ and thereby the solution $Z$ in $G$. Later, we will see that one of the recursive calls of {\sc Branching-Set} indeed correspond to $G'$.

\begin{lemma}\label{lemma:the-big-one}
Let $G_1 = G[R[W_1,P_1]]$ be a boundaried graph with $P_1^r$ as the boundary. There exists a $|P_1^r|$-boundaried graph $\hat{G}$ which is at most $k^{5(pd)^2}$ 
 in size with an annotated set of vertices $\Delta(\hat{G})$ of size at most $k$, and a bijection $\mu : \partial(\hat{G}) \rightarrow P_1^r$ such that the glued graph $G' = G_1 \otimes_\mu \hat{G}$ has the property that {\sc Branching-Set} procedure with input as $((G' \setminus \Delta(\hat{G}),|K^r|, W_1, P_1^{nr}, U \cup V(\hat{G}) \setminus P_1^r),|X^r|,|K^r \setminus X^r|)$ returns a set $\mathcal{R}'$ that intersects $K^r$.
\end{lemma}

\begin{proof}

\sloppy To develop the intuitions behind the proof, we first prove that for the graph $G'' = G[R[W_1, X]]$,  {\sc Branching-Set} with input as {$((G'' \setminus K^{nr},|K^r|, W_1, P_1^{nr}, U),|X^r|,|K^r \setminus X^r|)$} returns a set $\mathcal{R}''$ that intersects $K^r$.

Suppose $\mathcal{R}''$ does not intersect $K^r$. The set $\mathcal{R}''$ by definition intersects any \pionetopidmodulator $Z'' \subseteq V(G'') \setminus U$ of size $|K^r|$ for the graph $G'' \setminus K^{nr}$ containing an $(|K^r \setminus X^r|,U)$-important separator $X''$. 
 The set $K^r \subseteq V(G'') \setminus U$ is also a \pionetopidmodulator of size $|K^r|$ for the graph $G'' \setminus K^{nr}$ and it contains a $W_1-P_1^{nr}$ separator $X^r$. Hence, if we can prove that the set $X^r$ is a $(|K^r \setminus X^r|,U)$-important separator in the graph $G'' \setminus K^{nr}$, we are done.

Suppose this is not the case. Then there exists a  $(|K^r \setminus X^r|,U)$-good separator $X'' \subseteq V(G'') \setminus U$ in the graph $G'' \setminus K^{nr}$ that well-dominates $X^r$. Note that since $G''$ is the graph $G[R[W_1, X]]$, the set of vertices reachable from $W_1$ after deleting $X^r$ in the graph $G'' \setminus K^{nr}$ denoted by $R_{G'' \setminus K^{nr}}(W_1,X^r)$ is the set $R_{G \setminus K^{nr}}(W_1,X)$. If $X'' \neq X^r$, the set $R_{G \setminus K^{nr}}(W_1,X) = R_{G'' \setminus K^{nr}}(W_1,X^r)  \subset R_{G'' \setminus K^{nr}}(W_1,X'')$ which cannot happen. 

Note that the graph $G''$ can be viewed as the graph obtained by gluing two boundaried graphs $G_1$ and $G_2$ both having boundary $P_1^r$ where $G_1 = G[R[W_1, P_1]]$ and $G_2 = G[NR(W_1, P_1) \cap R(W_1, X) \cup P_1^r \cup K^{nr}]$ with the bijection being an identity mapping from $P_1^r$ into itself. Unfortunately the graph $G_2$ is not of size $k^{\mathcal{O}(1)}$ size and hence doesn't satisfy the conditions required for the lemma. We now aim to construct a graph $\hat{G}$ by keeping some $k^{\mathcal{O}(1)}$ vertices of $G_2$.

Let $V_2 = (NR(W_1,P_1) \cap R(W_1,X)) \cup P_1^r \cup K^{nr}$. The set $V_2 \setminus (P_1^r \cup K^{nr})$ contains the vertices which are disconnected by $P_1$ from $W_1$ but are not disconnected from $W_1$ by $X$.  We have $G_2 = G[V_2]$. \\

\noindent{\textbf{Marking Vertices of Forbidden Sets:}} 

We now perform the following marking scheme on the graph $G$ where we mark some vertices of $V_2$ to construct a smaller graph $G'$. Before this though, we need to define the following notations.

Let $p$ denote the size of the maximum sized subgraph present among all the families $\mathcal{F}_i$.
Let $\mathbb{H} = \{(H_1, \dotsc, H_d): H_i \in \mathcal{F}_i, i \in [d]\}$. For $\mathcal{H} = (H_1, \dotsc, H_d) \in \mathbb{H}$, let $\mathbb{B}_{\mathcal{H}} = \{(B_{1}, \dotsc, B_{d}) : B_i \subseteq V(H_i), i \in [d]\}$. Let $\mathbb{P}_1^r =\{(Q_{1}, \dotsc, Q_{d}) : Q_i \subseteq P_1^r, |Q_i| \leq p, i \in [d]\}$. 

Let $\mathbb{T}_{\mathcal{H}}$ be the collection of tuples  $(t_1, \dotsc, t_r)$ where $t_i$ is a pair of elements $t_i^1, t_i^2 \in P_1^r \cup \{\emptyset \}$ and $r = {\sum\limits_{i \in [d]}^{}|H_i| \choose 2}, i \in [r]\}$. We use the bijection $\rho :{\bigcup_{i \in [d]}H_i \choose 2} \rightarrow [r]$ so that $\rho(\{a,b\})$ denotes the index associated to the pair of vertices $a, b \in \bigcup_{i \in [d]}H_i$.

For all tuples $\langle \mathcal{H}, \mathcal{B}_{\mathcal{H}}, \mathcal{P}_1^r, \mathcal{T}_{\mathcal{H}} \rangle$ where $\mathcal{H} = (H_1, \dotsc, H_d) \in \mathbb{H},  \mathcal{B}_{\mathcal{H}}  = (B_{1}, \dotsc, B_{d}) \in \mathbb{B}_{\mathcal{H}}, \mathcal{P}_1^r = (Q_{1}, \dotsc, Q_{d}) \in \mathbb{P}_1^r$ and $\mathcal{T}_{\mathcal{H}} = (t_1, \dotsc, t_r) \in \mathbb{T}_{\mathcal{H}}$, if there exists a forbidden set $C \subseteq (V_2 \cup V(G_1))$ of the graph $G \setminus K^{nr}$ such that
\begin{itemize}
\item For all $i \in [d]$, there exists a subset $C_i \subseteq C$ such that $G[C_i]$ is isomorphic to $H_i$,
\item for sets $C_i^+ = V_2 \cap C_i$, we have graphs $G[C_i^+]$ isomorphic to  $H_i[B_i]$,
\item the set $P_1^r \cap C_i = Q_{i}$ and
\item for vertices $a_i \in C_i$ and $a_j \in C_j$ with $i,j \in [d]$, there is path $P'$ from $a_i$ to $a_j$ in the graph $G \setminus K^{nr}$ such that 
the first and last vertex of $P'$ in the set $P_1^r$ that has a neighbor to the set $R(W_1, P_1)$ is $t_{\rho(a_i,a_j)}^1$ and $t_{\rho(a_i,a_j)}^2$ respectively, (When the path $P'$ has only one such vertex $v$, we denote it by the pair $\{v,\emptyset\}$. If the path has no such vertex, then we denote it by the pair $\{\emptyset, \emptyset \}$, . Also note that the existence of such paths for all pair of vertices in $C$ shows that $G[C]$ is connected in the graph $G \setminus K^{nr}$).
\end{itemize}

then for one such forbidden set $C$, we mark the set  $C^+ = C \cap V_2$.  Let $M'$ be the set of vertices marked in this procedure. We call the corresponding forbidden sets $C$ as marked forbidden sets.

\begin{figure}[t]
\centering
	\includegraphics[scale=0.3]{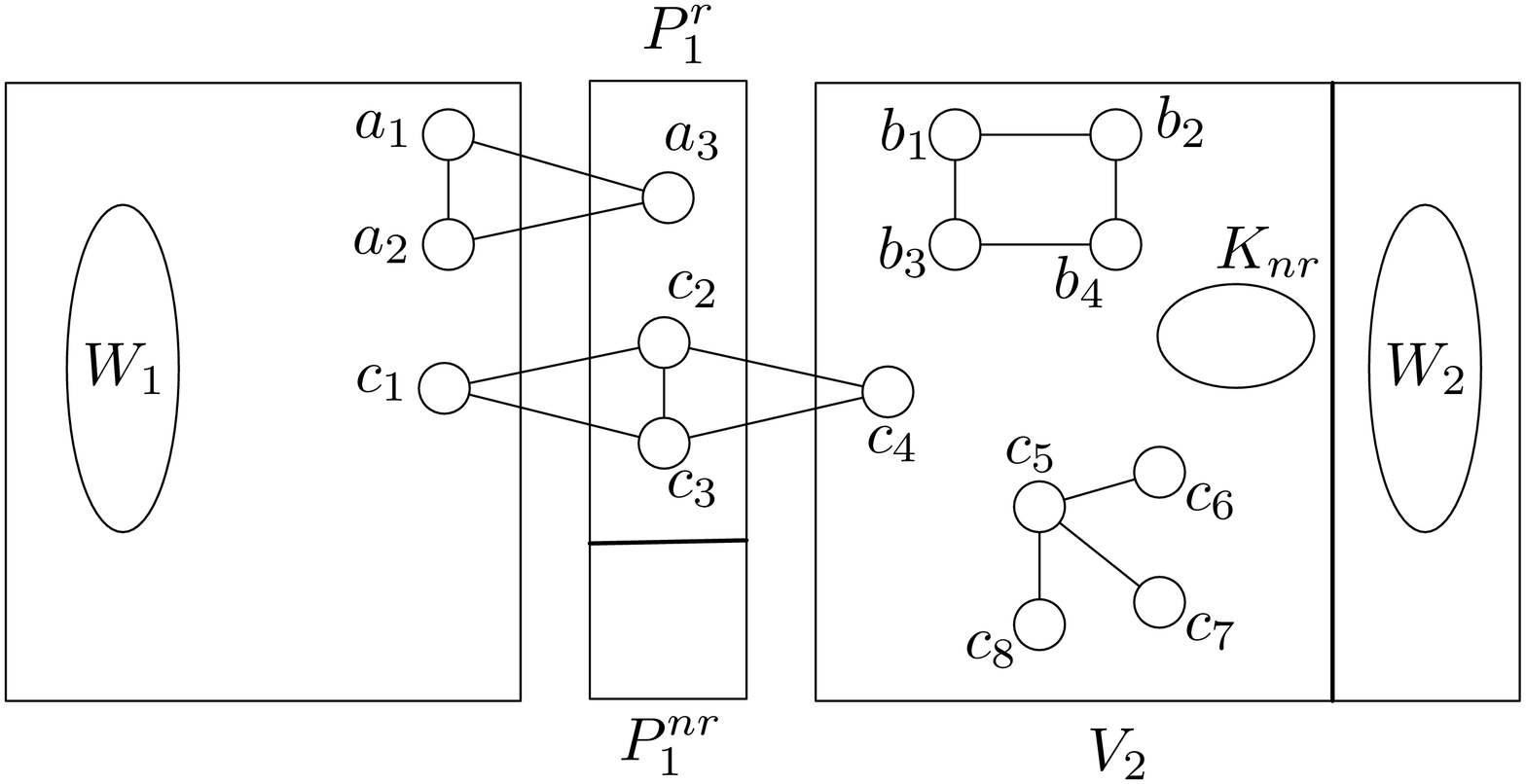}
	\caption{Example of a marked forbidden set}
\label{fig:marked-set}
\end{figure}

See figure \ref{fig:marked-set} for an example of a marked forbidden set. We have $\mathcal{H} = (H_1, H_2, H_3)$ where $H_1$ is a triangle, $H_2$ is a $C_4$ and $H_3$ has two connected components, one of which is $K_4$ after removal of an edge and the other is a claw $K_{1,3}$. The graph induced by the set of vertices $\{a_1, a_2, a_3\}$  is isomorphic to $H_1$, the one induced by the set of vertices $\{b_1, b_2, b_3, b_4\}$ is isomorphic to $H_2$ and the one induced by the set of vertices $\{c_1, c_2, c_3, c_4, c_5, c_6, c_7, c_8\}$ is isomorphic to $H_3$. We have $\mathcal{B}_{\mathcal{H}}  = (B_{1}, B_2 , B_{3})$ where $B_1$ is a singleton vertex, $B_2$ is the cycle graph $C_4$ and $B_3$ has two connected components, one of which is a triangle and the other is a claw $K_{1,3}$. Notice that there are the graphs induced by $H_1$, $H_2$ and $H_3$ when we restrict the set of vertices to $V_2$. We have $\mathcal{P}_1^r = (Q_{1}, Q_2 , Q_{3})$ as $(\{a_3\}, \{\emptyset\}, \{c_2, c_3\})$.

\begin{figure}[t]
\centering
	\includegraphics[scale=0.3]{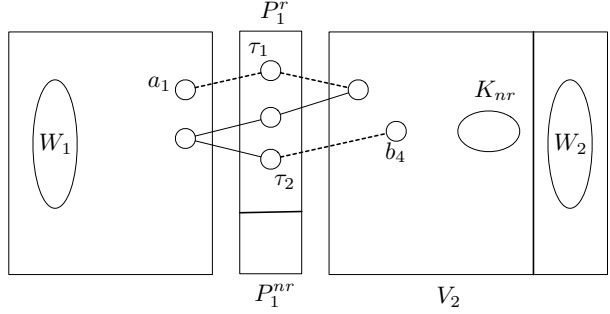}
	\caption{Example showing paths between vertices of a marked forbidden set}
\label{fig:marked-set-2}
\end{figure}

See figure \ref{fig:marked-set-2} where we look at a path between $a_1$ and $b_4$ of the same marked forbidden set. The first and last vertex of such a path is $\tau_1$ and $\tau_2$ respectively. Hence 
the entry of $\mathcal{T}_{\mathcal{H}}$ corresponding to the pair $(a_1, b_4)$ is $\{\tau_1, \tau_2\}$.


%

We now bound the size of $M'$. We know that each graph class $\mathcal{F}_i$ has a finite number of finite sized graphs. Let $f = \max_{i \in d} |\mathcal{F}_i|$. The size of $\mathbb{H}$ is the number of tuples $\mathcal{H} = (H_{1}, \dotsc, H_{d})$ which is at most $f^d$. Since $|H_i| \leq p$, the size of $\mathbb{B}_{\mathcal{H}}$ is bounded by the number of tuples $(B_{1}, \dotsc, B_{d})$ which is at most $2^{pd}$.  Since the set $Q_i$ is of size at most $p$, the size of $\mathbb{P}_1^r$  is bounded by $k^{(p+1)d}$. Each vertex in a pair in $\mathbb{T}_{\mathcal{H}}$ is a pair of vertices of $P_1^r$. The number of such pairs is bounded by $(k+1)^2$. Since $r = {\sum\limits_{i \in [d]}^{}|H_i| \choose 2} \leq {pd \choose 2}$, the size of $\mathbb{T}_{\mathcal{H}}$ is bounded by $((k+1)^2)^{{pd \choose 2}}$. Overall, we can conclude that the number of tuples $\langle \mathcal{H}, \mathcal{B}_{\mathcal{H}}, \mathcal{P}_1^r, \mathcal{T}_{\mathcal{H}} \rangle$ is at most $\eta = f^d 2^{pd} k^{(p+1)d}k^{2(pd)^2}$. For each of these tuples we mark the set $C \cap V_2$ which is of size at most $pd$. Hence we can conclude that $|M'| \leq \eta pd$. The same bound holds for the vertices corresponding to the marked forbidden sets which we denote by $M_F$. \\

\noindent{\textbf{Preserving Connectivity of the Marked Forbidden Sets:}} 
We now aim to keep some vertices in $V_2$ other than those in $M'$ so that for every marked forbidden set $C$, the graph $G[C]$ remains connected in the resulting graph. We also add the requirement that the connectivity between every vertex in $C^+$ and vertices in $P_1^r$ and also between pairs of vertices in $P_1^r$ in the graph $G[V_2]$ are preserved. Let $F$ be the forest of minimum size in the graph $G$ such that it satisfies these connectivity requirements. Note that $M' \subseteq V(F)$. 

We now try to bound the size of the forest $F$. Note that any leaf of the forest $F$ corresponds to some vertex in the marked forbidden set $M_F \cup P_1^r$. This is because it is not the case, then for some leaf vertex $u \in V(F)$, the forest $F - \{u\}$ also preserves the connectivities required for marked forbidden sets contradicting that $F$ is the forest of minimum size. Hence the number of leaves is bounded by $\eta pd + k$.

By properties of any forest, the number of vertices of degree 3 or more is at most the number of leaves. Hence such vertices of $F$ are also bounded by $\eta pd + k$. Hence it remains to bound the number of degree 2 vertices in $F$.

\begin{figure}[t]
\centering
	\includegraphics[scale=0.3]{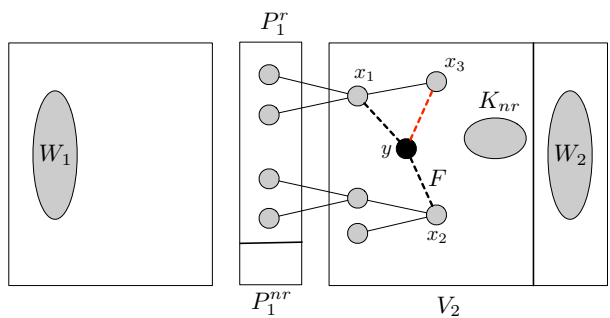}
	\caption{Forest $F$ that provides required connectivities of marked forbidden set vertices. The vertices colored grey correspond to marked vertices and white correspond to other vertices of $F$. The forest $F$ has a degree two path between $x_1$ and $x_2$ with all the internal vertices unmarked. If an unmarked vertex $y$ in this path has an edge to some marked vertex $x_3$, then the forest obtained by replacing an edge adjacent to $y$ in the path with $(x_3,y)$ also preserves the connectivities but has an unmarked vertex as a leaf giving a contradiction.}
\label{fig:forest-size-bound}
\end{figure}

Let us focus on a degree 2 path $P$ of $F$ with endpoints either a leaf of $F$ or degree at least 3 or any vertex in $M_F$ or $P_1^r$. Suppose $P$ has at least 3 internal vertices. We claim that all the internal vertices of $P$ except the first and last internal vertices are not adjacent to any vertices of $F$ in the graph $G$ induced on $V(F)$ other than its neighbors in the path $F$. Suppose this is not the case for some internal vertex $u$ with $u_1$ and $u_2$ being the two neighbors of $u$ in $P$. Hence $u$ is adjacent to some other vertex $v$ of $F$. Note that the edge $(u,v)$ is not in the forest $F$. Let us add this edge to the forest $F$ creating a unique cycle $C$ containing $(u,v)$. Without loss of generality, let $u_1$ be the other neighbor of $u$ in $C$. Let $F_1$ be the forest created by adding the edge $(u,v)$ and removing the edge $(u,u_1)$. Note that we now have a forest $F_1$ where $u_1$ is a leaf vertex that is not marked. Then  $F_1 - \{u_1\}$ is also a forest that preserves the connectivities that $F$ did with a fewer number of vertices. This contradicts that $F$ is the forest with the minimum number of vertices. See Figure \ref{fig:forest-size-bound} for an illustration regarding this proof.

Since $P$ does not have edges from the internal vertices to other vertices of $F$, we can contract these paths up to a certain length and preserve connectivities of $F$.


Let $G_2' = G[V(F)]$. Let $G_2$ be the graph obtained from $G_2'$ by contracting all the degree 2 paths in the graph of length more than $4pd+2$ to length $4pd+2$. Since the number of degree 2 paths in $F$ is bounded by $(2(|M_F| + |P_1|))^2$ and each path is bounded by size $4pd+2$, we have $V(G_2) \leq 2(\eta pd + k)^2 (4pd+2)$ which is at most $k^{5(pd)^2}$. \\ 

\noindent{\textbf{Construction of $G'$:}} 
Let $G'$ be the graph obtained by gluing the boundaried graphs $G[R[W_1, P_1]]$ and $G_2$ both having $P_1^r$ as the boundary with the bijection corresponding to the gluing being an identity mapping from $P_1^r$ to itself. See figure \ref{fig:smaller-graph}. We also similarly define $G'''$ as the graph obtained by gluing the graphs $G[R[W_1, P_1]]$ and $G_2'$ both having $P_1^r$ as the boundary with the bijection corresponding to the gluing being an identity mapping from $P_1^r$ to itself.

\begin{figure}[t]
\centering
	\includegraphics[scale=0.3]{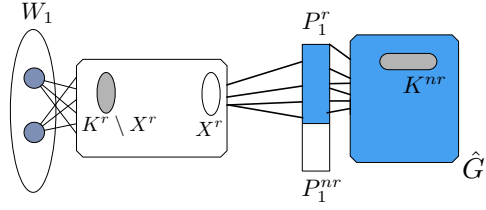}
	\caption{The graph $G' \setminus K^{nr}$ obtained from gluing the graphs $G[R[W_1,P_1]$ and $\hat{G} \setminus K^{nr}$ along $P_1^r$ where $K^r$ is an optimal {\pionetopidmodulator}}
\label{fig:smaller-graph}
\end{figure}

We claim that for the graph $G'$,  {\sc Branching-Set} with input $((G' \setminus K^{nr},|K^r|, W_1, P_1^{nr}, U'),|X^r|,|K^r \setminus X^r|)$ returns a set $\mathcal{R}'$ that intersects $K^r$ where $U' = U \cup V(\hat{G}) \setminus P_1^r$. Suppose not. 
The set $\mathcal{R}'$ by the definition of {\sc Branching-Set} procedure intersects any \pionetopidmodulator $Z' \subseteq (V(G') \setminus U')$ of size $|K^r|$ in the graph $G' \setminus K^{nr}$ containing an $(|K^r \setminus X^r|,U')$-important $W_1-P_1^{nr}$ separator. 
 Since $K$ is a \pionetopidmodulator in the graph $G[R[W_1,X]]$ with $K \cap U' = K^{nr}$, the set $K^r = K \setminus K^{nr}$ is a \pionetopidmodulator of size $|K^r|$ for the graph $G[R[W_1,X]] \setminus K^{nr}$ with no vertices from $U'$. 
Let us turn the focus to the graph $G'''$ where the degree $2$ paths are not contracted. Since $G''' \setminus K^{nr}$ is an induced subgraph of the graph $G[R[W_1,X]] \setminus K^{nr}$, the set $K^r$ is a \pionetopidmodulator of size $|K^r|$ for the graph $G''' \setminus K^{nr}$ as well. In other words, every connected component of the graph $G''' \setminus K$ belongs to at least one of the graph class $\Pi_i$ with $i \in [d]$.

\begin{claim}\label{claim-kr-mod}
$K^r$ is a \pionetopidmodulator of size $|K^r|$ for the graph $G' \setminus K^{nr}$. 
\end{claim}

 The graph $G' \setminus K$ can be viewed as obtained from  $G''' \setminus K$ by contracting some of the degree 2 paths of length more than $4pd+2$ to $4pd+2$. If we can prove that after doing this contraction in $G''' \setminus K$, the connected components still belongs to at least one of the graph class $\Pi_i$ with $i \in [d]$, we are done.
 
 Suppose not. Then there is a forbidden set $C$ in one of the connected components of the graph $G' \setminus K$. Let us now uncontract the edges that we contracted. We will show that in the resulting graph $G''' \setminus K$, there exists a forbidden set $C'$ which is isomorphic to $G[C]$. This contradicts our assumption that $G''' \setminus K$ is such that each of its components belongs to at least one of the graph classes $\Pi_i$ for $i \in [d]$.
 
 Let $C = \bigcup_{i \in [d]} C_i$ where $C_i$ isomorphic to the graph $H_i \in \FF_i$. Let $\alpha$ be one of the paths in the graph $G''' \setminus K$ with degree 2 vertices which was contracted to a path $\alpha'$ of length $4pd+2$ in the graph $G' \setminus K$. The graph induced by $C_i \cap V(\alpha')$ is such that each connected component is a path of length at most $p$. 
 In the path $\alpha$ too we can find a subset of vertices such that the graph induced by those vertices is isomorphic to the graph induced by $C_i \cap V(\alpha')$. 
Since $\alpha'$ has size $4pd+2> p$, no connected component of $C_i$ for any $i \in [d]$ has both the endpoints of $\alpha'$.  Hence, if we replace $C_i \cap V(\alpha')$ with the corresponding subsets we identified in $\alpha$, we get subsets $C_i'$ so that the set $C' = \bigcup_{i \in [d]} C_i'$ in the graph $G''' \setminus K$ is a forbidden set isomorphic to $C$. The connectivity of $C'$ is preserved as we only uncontract some edges. This contradicts that  $G''' \setminus K$ is such that each of its components belongs to at least one of the graph classes $\Pi_i$ for $i \in [d]$. This concludes the proof of the claim.

We now prove that $X^r$ is a $(|K^r \setminus X^r|,U')$-important separator in $G' \setminus K^{nr}$. This  implies that the set returned from the recursive procedure $\mathcal{R}'$ intersects $K^r$ completing the proof of the lemma.

\begin{claim}$X^r$ is a $(|K^r \setminus X^r|,U')$-important separator in $G' \setminus K^{nr}$.
\end{claim}

 
 


We first prove that the set $X^r$ is $(|K^r \setminus X^r|,U')$-good $W_1-P_1^{nr}$ separator in the graph $G''' \setminus K^{nr}$. We know that $X$ is a $W_1-P_1^{nr}$ separator in the graph $G$. Hence $X^r$ is $W_1-P_1^{nr}$ separator in the graph $G \setminus X^{nr}$. Since $G''' \setminus K^{nr}$ is an induced subgraph of $G \setminus X^{nr}$, we can conclude that $X^r$ is a $W_1-P_1^{nr}$ separator in the graph $G''' \setminus K^{nr}$. Since the graph $G' \setminus K^{nr}$ can be seen as obtained from $G''' \setminus K^{nr}$ by contracting some degree 2 paths with none of the edges with endpoints in $X^r$ contracted, $X^r$ is also a $W_1-P_1^{nr}$ separator in the graph $G' \setminus K^{nr}$.

Suppose $X^r$ is not a $(|K^r \setminus X^r|,U')$-good $W_1-P_1^{nr}$ separator in the graph $G' \setminus K^{nr}$. Then the graph $G' \setminus (K^{nr} \cup X^r)$ does not contain a  \pionetopidmodulator of size $|K^r \setminus X^r|$ with undeletable set $U'$. We claim that the set $K^r \setminus X^r$ is indeed such a \pionetopidmodulator . Suppose not. Then there exists a forbidden set $C$ in the graph $G' \setminus (K^{nr} \cup X^r \cup (K^r \setminus X^r)) = G' \setminus K$. But this contradicts Claim \ref{claim-kr-mod}.

Hence it remains to show that the set $X^r$ is a $(|K^r \setminus X^r|,U')$-important $W_1-P_1^{nr}$ separator in the graph $G' \setminus K^{nr}$. Suppose this is not the case.
Then there exists a  $(|K^r \setminus X^r|,U')$-good $W_1-P_1^{nr}$ separator $X'$ in the graph $G' \setminus K^{nr}$ that well-dominates $X^r$. We claim that if so, the set $\hat{X} = X' \cup X^{nr}$ is an $(\ell,U)$-good $W_1-W_2$ separator in the graph $G$ with undeletable set $U$ well-dominating $X$. This would contradict that $X$ is an $(\ell,U)$-important $W_1-W_2$ separator in the graph $G$.

Let $Y'$ be the set witnessing that $X'$ is a $(|K^r \setminus X^r|,U')$-good $W_1-P_1^{nr}$ separator in the graph $G' \setminus K^{nr}$. Note that $X' \cup Y'$ is contained in the set $R[W_1,P_1]$ as it cannot contain vertices from the set $(V(\hat{G}) \setminus P_1^r) \subseteq U'$.

We now claim that $Y' \cup (K^{nr} \setminus X^{nr})$ is the set witnessing that  $\hat{X}$ is an $(\ell,U)$-good $W_1-W_2$ separator in the graph $G$. Suppose this is not the case. Then there exists a forbidden set $C$ in the graph $G[R[W_1,\hat{X}]] \setminus (X' \cup X^{nr} \cup Y' \cup (K^{nr} \setminus X^{nr})) = G[R[W_1,\hat{X}]] \setminus K'$ where $K' = X' \cup Y' \cup K^{nr}$.




If $C \subseteq V(G') \setminus K'$, then the forbidden set $C$ occurs in the graph $G' \setminus K'$ contradicting that $X'$ is an $(|K^r \setminus X^r|,U')$-good $W_1-P_1^{nr}$ separator in the graph $G' \setminus K^{nr}$. Hence $C \cap ((R[W_1,\hat{X}] \setminus K') \setminus V(G'))$ is non-empty.
Let $C = \bigcup_{i \in d} C_i$ where $G[C_i]$ is isomorphic to graphs $H_i \in \FF_i$. Let $C_i^+ = C_i \cap NR[W_1,P_1]$ and $C^+ = \bigcup_{i \in d} C_i^+$. We have graphs $G[C_i^+]$ isomorphic to graphs $H_i[B_i]$ for subsets $B_i \subseteq V(H_i)$. Let $C_i^{P_1^r} = C_i \cap P_1^r$. Since $G[C]$ is connected, for vertices $a_i \in C_i$ and $a_j \in C_j$ with $i, j \in [d]$, there is a path $P_{a_i,a_j}$ from $a_i$ to $a_j$ in the graph $G[R[W_1,\hat{X}] \setminus K'$. Let the first and last vertices of $P_1^r$ in $P_{a_i,a_j}$ be the pair $t_{\rho(a_i,a_j)}$.
Then for the tuple  $\langle \mathcal{H}, \mathcal{B}_{\mathcal{H}}, \mathcal{P}_1^r, \mathcal{T}_{\mathcal{H}} \rangle$ where $\mathcal{H} = (H_1, \dotsc, H_d) \in \mathbb{H},  \mathcal{B}_{\mathcal{H}}  = (B_{1}, \dotsc, B_{d}) \in \mathbb{B}_{\mathcal{H}}, \mathcal{P}_1^r = (C_1^{P_1^r}, \dotsc, C_d^{P_1^r}) \in \mathbb{P}_1^r$ and $\mathcal{T}_{\mathcal{H}} = (t_1, \dotsc, t_{|C| \choose 2}) \in \mathbb{T}_{\mathcal{H}}$, there exists a forbidden set $C_M \subseteq (V_2 \cup V(G_1))$
 of the graph $G \setminus K^{nr}$ which is marked. 

Let $C_{M,i}$ be the set such that $G[C_{M,i}]$ is isomorphic to $H_i$. Also let $C_{M,i}^+ = C_{M,i} \cap (V(F))$. The set $C_M$ can be viewed as replacing the vertices $C^+$ of $C$ with $C_M^+$. 
 
We first claim that the set of vertices of $C_M$ is present in the graph $G' \setminus K'$. Recall that $G'$ is obtained by contracting some degree 2 vertices of $G'''$ which in turn is obtained by gluing the graphs $G[R[W_1, P_1]]$ and $G[V(F)]$. All the vertices of $C_M$ is present in the graph $G''' \setminus K^{nr}$ as it contains all the marked vertices. In particular, all the vertices of $C_M$ in the set $V_2$ are contained in the set $M'$, the set of vertices of all the marked forbidden sets contained in $V_2$. When we transform $G'''$ to $G'$, we only contract degree 2 paths in $V_2$ none of whose vertices belong to $M'$ and hence $C_M$. 
Hence $C_M$ is a present in the graph $G' \setminus K'$ as well.

If we can prove that $C_M$ is in a connected component of $G' \setminus K'$, we can conclude that $C_M$ is a forbidden set in the graph $G' \setminus K'$ contradicting that $X'$ is a $(|K^r \setminus X^r|,U')$-good $W_1-P_1^r$ separator in the graph $G' \setminus K^{nr}$.

Suppose $C_M$ is not in a connected component of $G' \setminus K'$. Then there exist a pair of vertices $u_1, u_2 \in C_M$ such that there is no path between $u_1$ and $u_2$ in the graph $G' \setminus K'$. But since $C_M$ corresponds to the tuple $\langle \mathcal{H}, \mathcal{B}_{\mathcal{H}}, \mathcal{P}_1^r, \mathcal{T}_{\mathcal{H}} \rangle$ with $\mathcal{T}_{\mathcal{H}} = (t_1 \dotsc t_{|C| \choose 2}) \in \mathbb{T}_{\mathcal{H}}$, there exists a path $P_{u_1,u_2}$ in the graph $G \setminus K^{nr}$ between $u_1$ and $u_2$ such that the first and last vertices of the path $P_{u_1,u_2}$ intersecting $P_1^r$ is the pair $t_{\rho(u_1, u_2)} = (\tau_1, \tau_2)$. 

Let us also look at vertices $u_1', u_2' \in C$ such that in the isomorphism from $C$ to $C_M$, $u_i$ is mapped to  $u_i'$ for $i \in \{1,2\}$. Since $C$ is connected, there is a path $P_{u_1',u_2'}$ between $u_1'$ and $u_2'$  in the graph $G[R[W_1,\hat{X}] \setminus K'$. Note that the forbidden sets $C$ and $C_M$ are both candidates for the marking procedure corresponding to the same tuple $\langle \mathcal{H}, \mathcal{B}_{\mathcal{H}}, \mathcal{P}_1^r, \mathcal{T}_{\mathcal{H}} \rangle$. Hence we can assume that the path $P_{u_1',u_2'}$ in the graph $G[R[W_1,\hat{X}] \setminus K'$ is such that the first and last vertices of the path $P_{u_1',u_2'}$ intersecting $P_1^r$ is the pair $(\tau_1, \tau_2)$. 

We now identify all the vertices of $P_1^r$ present in the path $P_{u_1',u_2'}$ and partition them accordingly. Specifically, let us partition the path $P_{u_1',u_2'}$ into a sequence of subpaths $\alpha_1, \dotsc , \alpha_q$ where the path $\alpha_1$ is from $u_1'$ to $\tau_1$, the path $\alpha_q$ is from $\tau_2$ to $u_2'$, the path $\alpha_i$ where $ 1 < i < q$ has its endpoints in $P_1^r$ and none of the internal vertices of the paths contain vertices of $P_1^r$.
We aim to use the path $P_{u_1,u_2}$ and the connectivities provided by the forest $F$ in $V_2$ to construct a path between $u_1$ and $u_2$ in the graph $G' \setminus K'$ leading to a contradiction.


Let us now look at the cases based on whether $u_i, u_i' \in R[W_1,P_1]$ or not.

\begin{itemize}
\item Case 1, $u_1, u_2 \in R[W_1,P_1]$: Note that since both the vertices are in $V(G')$, we have $u_i = u_i'$ for $i \in \{1,2\}$. The paths $\alpha_i$ which is not present in $G' \setminus K'$ are those whose internal vertices contains some vertices of $V_2 \setminus V(F)$. The paths $\alpha_1$ and $\alpha_q$ are present in $G' \setminus K'$ as all its internal vertices including $u_1$ and $u_2$ are not in $V_2$. Hence such paths $\alpha_i$ have both its endpoints in $P_1^r$. Also since $P_1^r$ separates $R(W_1,P_1)$ from $V_2 \setminus P_1^r$, all paths $\alpha_i$ with $1 < i < q$ are such that all its internal vertices are either in $R(W_1,P_1)$ or in $V_2 \setminus P_1^r$. The former kind of paths are also present in $G' \setminus K'$.

Hence the path $\alpha_i$ that contains vertices of $V_2 \setminus F$ are such that its endpoints are in $P_1^r$ and all its internal vertices in $V_2$. The forest $F$ preserved connectivities of vertices in $P_1^r$ within $V_2$ including the endpoints of $\alpha_i$. Hence we can replace $\alpha_i$ with the unique path between its endpoints in the forest $F$. Note that such a path is   disjoint from $K'$ and hence is present in $G' \setminus K'$.

By replacing all such paths $\alpha_i$ with those in $F$, we get a walk from $u_1$ to $u_2$ in the graph $G' \setminus K'$.

\item Case 2, $u_1 \in R[W_1,P_1], u_2 \in V_2 \setminus P_1^r $: In this case, we have $u_1 = u_1'$. But it could be the case that $u_2 \neq u_2'$. Also it could be that $u_2'$ is a vertex not in $G' \setminus K'$.

Like we did in the previous case, we could replace all the paths $\alpha_i$ to paths in $G' \setminus K'$ using the forest $F$. Hence we have a path from $u_1$ to $\tau_2$ in the graph $G' \setminus K'$. We now focus on the path $\alpha_q$. We know that there is a subpath from $\tau_2$ to $u_2$ in the path $P_{u_1,u_2}$ in the graph $G \setminus K^{nr}$. Since $\tau_2$ is the last vertex of $P_1^r$ in the path, we know that this subpath is contained in the set $V_2$. Since $F$ is a forest that preserved connectivities of vertices between $M'$ and $P_1^r$ in the set $V_2$, there is a path from $\tau_2$ to $u_2$ in the forest $F$. Again note that such a path is   disjoint from $K'$ and hence is present in $G' \setminus K'$. We replace $\alpha_q$ with this path in $F$ to get a walk from $u_1$ to $u_2$ in the graph $G' \setminus K'$.

\item Case 3, $u_2 \in R[W_1,P_1], u_1 \in V_2 \setminus P_1^r $: This case is symmetric to the previous case and the proof goes accordingly.

\item Case 4, $u_1, u_2 \in V_2 \setminus P_1^r $ : In this case,  it could be that $u_2 \neq u_2'$ and $u_2 \neq u_2'$. Also it could be that the vertices $u_1'$ and $u_2'$ are not in $G' \setminus K'$.

We replace the path $\alpha_1$ to one in the forest as we done for $\alpha_q$ in Case 2. Since $\tau_1$ is the first vertex of $P_1^r$ in the subpath between $u_1$ and $\tau_1$ in the path $P_{u_1, u_2}$, such a path is contained in the set $V_2$. The forest $F$ preserves the connectivity between $u_1$ and $\tau_1$ in $V_2$.  Hence we can replace the path $\alpha_1$ with the one in $F$ which is disjoint from $K'$.

The other paths $\alpha_i$ are replace similarly as in Case 2 to get a walk from $u_1$ to $u_2$ in the graph $G' \setminus K'$.

\end{itemize}

%


\begin{figure}[t]
\centering
	\includegraphics[scale=0.3]{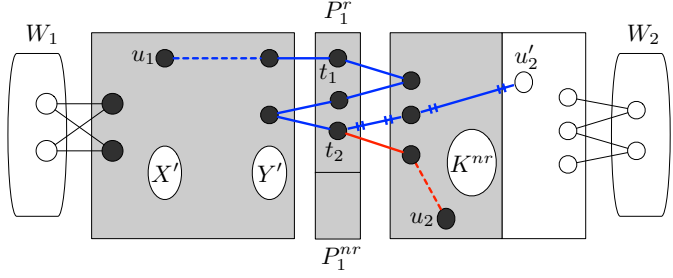}
	\caption{A demonstration of how $u_1, u_2 \in C_M$ are connected in the graph $G' \setminus K'$ (denoted by the grey region). We know from the marking procedure that both $C$ and $C_M$ are such that paths between vertices corresponding to $u_1$ and $u_2$ have its first and last vertex of $P_1^r$ as $t_1$ and $t_2$. We replace the path between $t_2$ and $u_2'$ with the path between $t_2$ and $u_2$ guaranteed from the forest $F$.}
\label{fig:new-separator-sequence}
\end{figure}

 Hence the graph $G[C_M]$ is connected in the graph $G' \setminus K'$. Hence $C_M$ is a forbidden set in the graph $G' \setminus K'$ contradicting that $X'$ is $(|K^r \setminus X^r|, U')$-good $W_1-P_1^{nr}$ separator in the graph $G' \setminus K^{nr}$. Hence $|X^r|$ is $(|K^r \setminus X^r|, U')$-important $W_1-P_1^{nr}$ separator in the graph $G' \setminus K^{nr}$. This concludes the proof the claim and thereby the lemma.

\end{proof}

From Lemma \ref{lemma:the-big-one}, we can conclude that there exists a $|P_1^r|$-boundaried graph $\hat{G}$ with an annotated set $\tilde{S}$ and an appropriate bijection $\mu : \partial(\hat{G}) \rightarrow P_1^r$ with the properties claimed in the statement of Lemma \ref{lemma:the-big-one}. Now consider the recursive instance of {\sc Branching-Set} with input $\langle(G_{P_i^r, \delta} \setminus \tilde{S}, k_1, W_1, P_i^{nr}),\lambda',\ell'\rangle$ where $G_{P_i^r, \delta}$ is the graph obtained by gluing together $G[R[W_1,P_1]]$ and $\hat{G}$ via a bijection $\mu$ 
$, \lambda' = |X^r|, k_1 = |K^r|$ and $\ell' = |K^r \setminus X^r|$.

To apply induction hypothesis on the above tuple, we first show that the tuple is valid. We show this by showing that $(G_{P_i^r, \delta} \setminus \tilde{S}, k_1, W_1, P_1^{nr})$ is a valid instance of \disjointpionetopiddeletioncomp. For this we need that $W_1 \cup P_1^{nr}$ is a 
\pionetopidmodulator for the graph $G_{P_i^r, \delta} \setminus \tilde{S}$ which is true as $W_1$ itself is such a set. Hence the tuple is valid and we can apply the induction hypothesis.

Since $X$ is $\ell$-important, from Lemma \ref{lemma:the-big-one}, it follows that $X^r$ must also be $k_1$-important in the graph $G_{P_i^r, \delta} \setminus  \tilde{S}$. By induction hypothesis, the tuple returns a set $\mathcal{R}'$ which intersects $K^r$. Since $K^r \subseteq Z$, we can conclude that $\mathcal{R'}$ intersects $Z$ as well. This completes the correctness of the {\sc Branching-Set} procedure.

\noindent
{\bf Bounding the set $\mathcal{R}$:} Let us look at the recursion tree of the {\sc Branching Set} procedure. The value of $\lambda$ drops at every level of the recursion tree. Since $\lambda \leq k$, the depth of the tree is bounded by $k$. The number of branches at each node is at most $k^3 \cdot 2^k \cdot k! \cdot 2^{k^{5(pd)^2}}$ ($k^3$ for choice of $\lambda', j$ and $\ell'$, $2^k$ for choice of $P_i^r$, $k!$ for the choice of the bijection $\delta$ and $ 2^{k^{5(pd)^2}}$ for the size of $\mathcal{H}$). Since, at each internal node, we add at most $2k$ vertices (corresponding to $P_1 \cup P_2$), we can conclude that the size of $\mathcal{R}$ is bounded by $ 2^{k^{5(pd)^2+2}}$.
Let $\mathcal{R}_{\lambda',\ell'}$ denote the set returned by {\sc Branching-Set} procedure with input $(\II, \lambda', \ell')$. We define $\mathcal{R}$ as the union of the sets $\mathcal{R}_{\lambda',\ell'}$ for all possible values of $\lambda'$ and $\ell'$. After this, in the {\sc Main-Algorithm} procedure we simply branch on every vertex $v$ of $\mathcal{R}$ creating new instances $(G - v,k-1,W_1,W_2)$ of {\disjointpionetopiddeletioncomp}. If $k<0$, we return NO. If Reduction Rule \ref{rr2} applies, we use it to reduce the instance. If this results in a non-separating instance with $W=W_1 \cup W_2$
, we apply the algorithm in Lemma \ref{lemma:non-sep} to solve the instance. Else we recursively run  {\sc Main-Algorithm} on the new instance. 

\noindent
{\bf Bounding running time of {\sc Main-Algorithm}:} We now bound the running time $T(k)$ for {\sc Main-Algorithm}. We ignore the constants in the polynomial of $k$ in the size of $\mathcal{R}$ and the nodes in the search tree for an easier analysis to show that the running time is $2^{poly(k)} n^{\OO(1)}$.

The depth of the branching tree is bounded by $k$ and the branching factor at each node is $|\mathcal{R}| \leq 2^{k^{5(pd)^2+2}}$. The time taken at each node is dominated  by the time taken for the procedure {\sc Branching-Set} corresponding to this node instance. Let $Q(k)$ denote the time taken for {\sc Branching-Set}. We have $T(k) = 2^{k^{\OO(1)}} T(k-1) + Q(k)$. 
Let us focus on the search tree for {\sc Branching-Set}. We know that the depth of the tree is bounded by $k$ and the branching factor is bounded by $ 2^{k^{\OO(1)}}$. The time spent at each node is dominated by algorithm of to enumerate graphs of size at most $2^{k^{5(pd)^2}}$ and 
that of the at most $\log n$ many sub-instances of {\sc Main-Algorithm} called with strictly smaller values of $k$, which is bounded by $2^{k^{\OO(1)}}n^{O(1)}+ \log n T(k-1)$. Hence overall we have $Q(k) = 2^{k^{\OO(1)}}Q(k-1) + 2^{k^{\OO(1)}}n^{O(1)}+ \log n T(k-1)$.

We now prove by induction that $T(k) = (2^{k^{\OO(1)}} + \log n)^kn^{\OO(1)}$ and $Q(k) = (2^{k^{\OO(1)}} + \log n)^kn^{\OO(1)}$. The base case is true as $T(1) =Q(1) = n^{\OO(1)}$. Assume the statement holds true for $2 \leq i \leq k-1$. Substituting the values for $T(k-1)$ and $Q(k-1)$ in the recurrence for $Q(k)$, we have
\begin{eqnarray*}
Q(k) & = & 2^{k^{\OO(1)}} (2^{k^{\OO(1)}} + \log n)^{k-1}n^{\OO(1)} + 2^{k^{\OO(1)}}n^{\OO(1)} \\
& & + \log n (2^{k^{\OO(1)}} + \log n)^{k-1}n^{\OO(1)} \\
& = & (2^{k^{\OO(1)}} + \log n)^{k-1}n^{\OO(1)} (2^{k^{\OO(1)}} +1 + 2^{k^{\OO(1)}} + \log n) \\
& = & (2^{k^{\OO(1)}} + \log n)^{k}n^{\OO(1)}
\end{eqnarray*}
Substituting the values for $T(k-1)$ and $Q(k)$ in the recurrence for $T(k)$, we have
\begin{eqnarray*}
T(k) & = & 2^{k^{\OO(1)}} (2^{k^{\OO(1)}} + \log n)^{k-1}n^{\OO(1)} + (2^{k^{\OO(1)}} + \log n)^{k}n^{\OO(1)} \\
& = & (2^{k^{\OO(1)}} + \log n)^{k-1}n^{\OO(1)} (2^{k^{\OO(1)}} + (2^{k^{\OO(1)}} + \log n))  \\
& = &(2^{k^{\OO(1)}} + \log n)^{k}n^{\OO(1)}
\end{eqnarray*}
Expanding the term $(2^{k^{\OO(1)}} + \log n)^{k}$ and observing that $(\log n)^k \leq (k \log k)^k+n $, we can conclude that $T(k) = 2^{k^{\OO(1)}}n^{\OO(1)}$.


\begin{lemma} \disjointpionetopiddeletioncomp\ can be solved in $2^{k^{\OO(1)}}n^{O(1)}$ time. 
\end{lemma}

\begin{proof}
Let $(G,k,W)$ be the instance of \disjointpionetopiddeletioncomp. We first apply Lemma \ref{lemma:non-sep} to see if there is a non-separating solution for the instance. If not, we branch over all $W_1 \subset W$ and for each such choice of $W_1$, 
apply {\sc Main-Algorithm} procedure with input $(G,k,W_1, W_2, \emptyset)$ to check if $(G,k,W_1, W_2 = W \setminus W_1)$ has a solution containing a $W_1 - W_2$ separator. The correctness and running time follows from those of Lemma \ref{lemma:non-sep} and correctness of the {\sc Main-Algorithm} procedure. 
\end{proof}


As mentioned in Section \ref{section:iterative-compression}, the time taken to solve \pionetopiddeletion\ is $2^{k+1} \cdot 2^{k^{\OO(1)}}n^{\OO(1)} = 2^{k^{\OO(1)}}n^{\OO(1)}$. This proves Theorem \ref{theorem:main-result}.

%

\section{Conclusion}

We have initiated a study on vertex deletion problems to scattered graph classes and showed that the problem is FPT when there are a finite number of graph classes, the deletion problem corresponding to each of the finite classes is known to be FPT and the properties that a graph belongs to each of the classes is expressible in CMSO logic.  Furthermore, we show that in the case where each graph class has a finite forbidden set, the problem is fixed-parameter tractable by a $O^*(2^{k^{\OO(1)}})$ algorithm. The existence of a polynomial kernel for these cases are natural open problems. A later paper by a subset of authors \cite{jacob2021faster} gives faster algorithms when the problem is restriced to a pair of some specific graph classes.






\end{document}